\documentclass{lmcs}
\pdfoutput=1

\usepackage{lastpage}
\renewcommand{\dOi}{14(4:17)2018}
\lmcsheading{}{1--\pageref{LastPage}}{}{}%
{Oct.~31,~2017}{Nov.~20,~2018}{}

\usepackage{amsmath}%
\usepackage{amssymb}%
\usepackage{amsthm}%
\usepackage{graphicx}%
\usepackage{hhline}%
\usepackage{subcaption}%
\usepackage{todonotes}%
\usepackage{trsym}%
\usepackage{wrapfig}%
\input xy
\xyoption{all}%

\newcommand{\adh}{\raisebox{-1.25pt}{$\TransformHoriz$}}		
\newcommand{\Adh}{{\rm Ad}}						
\DeclareMathOperator{\AfH}{aff}						
\newcommand{\after}{\mathrel{\circ}}					
\newcommand{\Alg}[1]{\mathbb{#1}}					
\newcommand{\Cat}[1]{\ensuremath{\mathbf{#1}}}				
\DeclareMathOperator{\CoH}{conv}					
\newcommand{\CSS}{\mathcal{P}_{{\rm c}}}				
\newcommand{\Dis}{\mathcal{D}{\mkern-3.5mu}_f}				
\newcommand{\EM}{\mathcal{E}{\kern-.2ex}\mathcal{M}}			
\DeclareMathOperator{\Ext}{Ext}						
\DeclareMathOperator{\id}{id}						
\newcommand{\nadh}{\mkern2mu\not\mkern-12mu\adh}			
\newcommand{\Pow}{\mathcal{P}{\mkern-4mu}_f}				
\newcommand{\Po}{\mathcal{P}	}			
\newcommand{\Pri}{{\rm P}}						
\DeclareMathOperator{\relint}{ri}					
\newcommand{\Sets}{\Cat{Sets}}						
\DeclareMathOperator{\ViH}{Vis}						

\newcommand{\qedd}{}%
\begin{document}

\title{Termination in Convex Sets of Distributions}

\author{Ana Sokolova}
\address{University of Salzburg, Austria}
\email{ana.sokolova@cs.uni-salzburg.at}

\author{Harald Woracek}
\address{TU Vienna, Austria}
\email{harald.woracek@tuwien.ac.at}

\keywords{convex algebra, one-point extensions, convex powerset monad}

\begin{abstract}
	Convex algebras, also called (semi)convex sets, are at the heart of modelling probabilistic systems including
	probabilistic automata. Abstractly, they are the Eilenberg-Moore algebras of the finitely supported distribution monad.
	Concretely, they have been studied for decades within algebra and convex geometry. 
	
	In this paper we study the problem of extending a convex algebra by a single point. 
	Such extensions enable the modelling of termination in probabilistic systems.
	We provide a full description of all possible extensions for a particular class of convex algebras: For a fixed convex
	subset $D$ of a vector space satisfying additional technical condition, we consider the algebra of convex subsets of $D$. 
	This class contains the convex algebras of convex subsets of distributions, modelling (nondeterministic) probabilistic 
	automata. We also provide a full description of all possible extensions for the class of free convex algebras, 
	modelling fully probabilistic systems. Finally, we show that there is a unique functorial extension, the so-called 
	black-hole extension.
\end{abstract}

\maketitle


\section{Introduction}
\label{sec:intro}

In this paper we study the question of how to extend a convex algebra by a single element. Convex algebras have been studied for many decades in the context of algebra, vector spaces, and convex geometry, see e.g.~\cite{stone:1949,flood:1981,gudder.schroeck:1980} and from a categorical viewpoint, see e.g.~\cite{fritz:2015,swirszcz:1974,roehrl:1994,pumpluen.roehrl:1995,boerger.kemper:1996}. Recently they have attracted more attention in computer science as well, see e.g.~\cite{doberkat:2006,jacobs:2010,jacobs.westerbaan.westerbaan:2015}.
One reason is that probability distributions, the main ingredient for modelling various probabilistic systems, see e.g.~\cite{VR99:tcs,BSV04:tcs,Sokolova11}, have a natural convex algebra structure. Even more than that, the set of distributions over a set $S$ carries the free convex algebra over $S$.
As a consequence, on the concrete side, convexity has notably appeared in the semantics of probabilistic systems, in particular probabilistic automata~\cite{SL94,Seg95:thesis,Mio14,HermannsKK14}. One particularly interesting development in the last decade in the theory of probabilistic systems is to consider probabilistic automata as transformers of \emph{belief states}, i.e., probability distributions over states, resulting in semantics on distributions, see~\cite{HermannsKK14,CPP09,CR11,DGHM08,DengGHM09,DengH13,PS07} to name a few. Convexity is inherent to this modelling and the resulting semantics that we call \emph{distribution bisimilarity}.

Additionally, on the abstract side, coalgebras over (categories of) algebras have attracted significant attention~\cite{SBBR10,JacobsSS15}. They make explicit the algebraic structure that is present in (the states of) transition systems and allow for its utilisation in the notion of semantics.  For coalgebras over convex algebras, the most important observation is that convex algebras are the Eilenberg-Moore algebras of the finitely supported distribution monad~\cite{swirszcz:1974,doberkat:2006,doberkat:2008,jacobs:2010}.  The first author, with coauthors, has recently studied the abstract coalgebraic foundation of probabilistic automata as coalgebras over convex algebras in~\cite{BSS17}, by providing suitable functors and monads on the category $\EM(\Dis)$ of Eilenberg-Moore algebras that model probabilistic automata. As a result, one gets a neat generic treatment and understanding of distribution bisimilarity.

One contribution of~\cite{BSS17} is identifying a convex-powerset monad on $\EM(\Dis)$ that together with a constant-exponent functor can be used to model probabilistic automata as coalgebras over $\EM(\Dis)$. However, the convex-powerset monad allows only for \emph{nonempty} convex subsets, and hence there is no notion of termination. As a consequence, one can only model \emph{input enabled} probabilistic automata. Hence, the question arises of how to add termination. One obvious way is to add termination that rules over any other behaviour: Consider a probabilistic automaton with two states $s$ and $t$; a distribution $ps + \bar pt$ with $\bar p=1-p$ 
over states $s$ and $t$ terminates if and only if one of the states terminates. We refer to this approach as \emph{black-hole} termination. Several distribution-bisimilarity semantics in the literature disagree on the treatment of termination, see e.g.~the discussion in~\cite{HermannsKK14} as well as~\cite{DGHM08,DengGHM09,DengH13,CPP09}. Understanding termination in probabilistic automata as transformers of belief states is the motivation for this work. On the level of convex algebras, termination amounts to the question of extending a convex algebra by a single element. 

Stated algebraically, the questions we address in this paper are: 
\begin{enumerate}
\item Given a convex algebra $\Alg X$, is it possible to extend it by a single point?
\item If yes, what are all possible one-point extensions? 
\item Which one-point extensions are functorial, i.e., do they provide a functor on $\EM(\Dis)$ that could then be used for modelling probabilistic automata as coalgebras over $\EM(\Dis)$? 
\end{enumerate}

Observe that extensions by a single point are different from the coproduct $\mathbb X+1$ in $\EM(\mathcal D_f)$; the coproduct was concretely described in \cite[Lemma~4]{jacobs.westerbaan.westerbaan:2015}, and 
	it has a much larger carrier than the {set} $X  + 1$.

Despite a large body of work on convex algebras, to the best of our knowledge, the problem of extending a convex algebra by a single element has not been studied, except for the black-hole extension mentioned above, see~\cite{fritz:2015}.

We answer the stated questions, and in particular our answers and main results are:
 
\begin{enumerate}
\item Yes, it is possible and there are many possible extensions in general. One of them is the mentioned black-hole extension. 
\item  We describe all possible extensions for the free convex algebra $\Alg D_S$ of probability distributions over a set $S$, see
Theorem~\ref{thm:free-extensions} in Section~\ref{sec:free}. Furthermore, we describe all possible extensions of an algebra $\CSS
\Alg D$ for $\Alg D$ being a convex subset of a vector space, satisfying a boundedness condition, see
Theorem~\ref{thm:Pc-extensions} in Section~\ref{sec:PcD}. As $\Alg D_S$ is a particular subset of a vector space, we get a
description of all possible extensions of $\CSS \Alg D_S$ which is exactly what is needed to understand termination in convex sets
of distributions.  \item We prove that only the black-hole extension is functorial, see Theorem~\ref{thm:functor} in
Section~\ref{sec:free}. 
\end{enumerate}

In addition, we provide many smaller results, observations, and examples that add to the vast knowledge on convex algebras. 

We mention that reading our results and proofs in detail does not require any prior knowledge beyond basics of algebra, with two
notable exceptions: (1) We do use some topological and geometric arguments in order to prove claims for the
construction of some of our examples; (2) We add small remarks about coalgebras and categories as we already did in this
introduction, assuming that readers are familiar with these basic notions (or will otherwise ignore the remarks made). 

This paper is an extended version of~\cite{SW17} including all proofs and additional examples (Section~\ref{sec:appendix-ex} and Section~\ref{sec:appendix-examples}).

\section{Convex Algebras}
\label{sec:CA}

By $\mathcal{C}$ we denote the signature of convex algebras
$$\mathcal{C} = \{ (p_i)_{i=0}^n \mid n \in \mathbb{N}, p_i \in [0,1], \sum_{i=0}^n p_i = 1 \}.$$
Intuitively, the $(n+1)$-ary operation symbol $(p_i)_{i=0}^n$ will be interpreted by a convex combination with coefficients $p_i$
for $i=0, \dots, n$.  For a real number $p \in [0,1]$ we set $\bar{p} = 1-p$. 

\begin{defi}\label{def:CA}
	A \emph{convex algebra} $\Alg{X}$ is an algebra with signature $\mathcal{C}$, i.e., a set $X$ together with an 
	operation $\sum_{i=0}^n p_i (-)_i$ for each operational symbol $(p_i)_{i=0}^n \in \mathcal{C}$, 
	such that the following two axioms hold:
	\begin{itemize}
		\item Projection: $\sum_{i=0}^n p_ix_i = x_j$ if $p_j = 1$.
		\item Barycenter: $\sum_{i=0}^n p_i \left(\sum_{j=0}^m q_{i,j} x_j\right) = 
			\sum_{j=0}^m \left( \sum_{i=0}^n p_i q_{i,j}\right) x_j$. 
\qedd
	\end{itemize}
\end{defi} 

We remark that the terminology in the literature is far from uniform. To give a few examples, 
convex algebras are called ``convex modules in \cite{pumpluen.roehrl:1995}, ``positive convex structures'' 
in \cite{doberkat:2006} (where $X$ is taken to be endowed with the discrete topology), 
``finitely positively convex spaces'' in \cite{wickenhaeuser:1988}, and ``sets with a convex structure'' 
in \cite{swirszcz:1974}. 

Convex algebras are the Eilenberg-Moore algebras of the finitely-supported distribution monad $\Dis$ on 
$\Sets$, cf.\ \cite[4.1.3]{swirszcz:1974} and \cite{semadeni:1973}, see also \cite{doberkat:2006,doberkat:2008} or 
\cite[Theorem~4]{jacobs:2010} where a concrete and simple proof is given. 
A convex algebra homomorphism is a morphism in the Eilenberg-Moore category $\EM(\Dis)$. 
Concretely, a convex algebra homomorphism $h$ from $\Alg{X}$ to $\Alg{Y}$ is a \emph{convex} (synonymously, \emph{affine}) map, 
i.e., $h\colon X \to Y$ with the property $h\left(\sum_{i=0}^n p_i x_i \right) = \sum_{i=0}^n p_i h(x_i)$.

\begin{rem}\label{rem:bin-suff}
	Let $\Alg{X}$ be a convex algebra. Then (for $p_n \neq 1$)
	\begin{equation}\label{eq:bin-suff}
		\sum_{i=0}^n p_i x_i = \overline{p_n} \left( \sum_{j=0}^{n-1} \frac{p_j}{\overline{p_n}} x_j\right) + p_n x_n
		.
	\end{equation}	
	Hence, an $(n+1)$-ary convex combination can be written as a binary convex combination using an $n$-ary convex
	combination. As a consequence, if $X$ is a set that carries two convex algebras $\Alg{X}_1$ and $\Alg{X}_2$ with
	operations $\sum_{i=0}^n p_i (-)_i$ and $\bigoplus_{i=0}^n p_i (-)_i$, respectively (and binary versions $+$ and $\oplus$,
	respectively) such that $px + \bar p y = px \oplus \bar p y$ for all $p, x, y$, then $\Alg X_1 = \Alg X_2$. 
\end{rem}

One can also see Equation~(\ref{eq:bin-suff}) as a definition --  the classical definition of Stone~\cite[Definition~1]{stone:1949}.
We make the connection explicit with the next proposition.

\begin{prop}\label{prop:bin-suff}
	Let $X$ be a set with binary operations $px + \bar py$ for $x, y \in X$ and $p \in (0,1)$. 
	Assume 
	\begin{itemize}
		\item Idempotence: $px + \bar p x = x$ for all $x\in X, p \in (0,1)$.
		\item Parametric commutativity: $px + \bar py = \bar p y + p x$ for all $x,y \in X, p \in (0,1)$.
		\item Parametric associativity: $p(qx + \bar qy) + \bar p z = 
			pqx + \overline{pq}\left( \frac{p\bar q}{\overline{pq}}y + \frac{\bar p}{\overline{pq}}z\right)$ 
			for all $x,y,z \in X, p,q, \in (0,1)$.  
	\end{itemize}
	Define $n$-ary convex operations inductively by the projection axiom and the formula~(\ref{eq:bin-suff}). Then $X$ becomes
	a convex algebra.
\end{prop}

\begin{proof}
	The proof is carried out by induction along the lines of \cite[Lemma~1--Lemma~4]{stone:1949}.
\end{proof}

This allows us to focus on binary convex combinations whenever more convenient. 

\begin{defi}\label{def:convex-set}
	Let $\Alg X$ be a convex algebra, and $C \subseteq X$. Then $C$ is called \emph{convex} if 
	it is the carrier of a subalgebra of $\Alg X$, i.e., if $px + \bar p y \in C$ for all $x, y \in C$ and $p \in (0,1)$.
\qedd
\end{defi}

Given convex algebras $\Alg X_i$, $i\in I$, the direct product $\prod_{i\in I}X_i$ is a convex algebra 
$\prod_{i\in I}\Alg X_i$ with operations defined componentwise. 
We call a relation on the carrier of a convex algebra $\Alg X$ a \emph{convex relation}, 
if it is a convex subset of $\Alg X\times\Alg X$. 

\begin{defi}\label{def:cancel}
	Let $\Alg X$ be a convex algebra. An element $z \in X$ is $\Alg X$-\emph{cancellable} if 
	$$\forall x,y \in X. \,\,\forall p \in (0,1). \,\,px + \bar pz = py + \bar pz \Rightarrow x=y.$$
	The convex algebra $\Alg X$ is \emph{cancellative} if every element of $X$ is $\Alg X$-cancellable. 
\qedd
\end{defi}

\begin{defi}\label{def:adheres}
	Let $\Alg X$ be a convex algebra. An element $x \in X$ \emph{adheres to} an element $y \in X$, 
	notation $x \adh y$, if $px + \bar p y = y$ for all $p \in (0,1)$. 
\qedd
\end{defi}

Observe that for a cancellative algebra the adherence relation equals the identity relation. 
As examples show, the converse need not hold.
The following simple properties of adherence will be needed on multiple occasions.

\begin{lem}\label{lem:adh-props}
Let $\Alg X$ be a convex algebra. The following properties hold.
\begin{enumerate}
 	\item For all $x, y \in X$, $x \adh y$ if and only if $px + \bar p y = y$ for some $p \in (0,1)$.
 	\item The adherence relation is reflexive and convex.
	\item For all $x, y \in X$, if $x \adh y$ then $pz  + \bar px \adh pz + \bar py$ for all $z \in X$ and $p \in (0,1)$. 
	\item If $z$ is $\Alg X$-cancellable, then for all $x, y \in X$ and $p \in (0,1)$
		$$
			pz  + \bar px \adh pz + \bar py \,\,\,\Rightarrow\,\,\, x \adh y
			.
		$$
\end{enumerate}    
\end{lem}

\begin{proof}
\hfill
	\begin{enumerate}
	\item Let $x, y \in X$.
		Consider the map $\varphi\colon [0,1] \to \Alg X$ defined by $\varphi(p) = px + \bar py$.
		Easy calculations show that 
		\begin{equation}\label{eq:ccstar1}
			(qp + \bar qr)x + \overline{(qp + \bar q r)}y = q(px+\bar py)+\bar q (rx + \bar r y) 
			,
		\end{equation}
		showing that $\varphi$ is convex. 
		The implication $\Rightarrow$ trivially holds. For the implication $\Leftarrow$ assume that $rx + \bar ry = y$ for some 
		$r \in (0,1)$. Then $\varphi(0) = y = \varphi(r)$ showing that the kernel of $\varphi$ 
		is a congruence of $[0,1]$ which is not the diagonal. Recall that a congruence on $[0,1]$ is an equivalence $R$ with the property that for all $p \in [0,1]$ and $(x_1, y_1), (x_2,y_2) \in R$, also $(px_1 + \bar p x_2, py_1 + \bar p y_2) \in R$. Equivalently, congruences are kernels of homomorphisms. By \cite[Lemma~3.2]{flood:1981}, 
		$\varphi$ is constant on $(0,1)$. 
		This shows that for all $p \in (0,1)$, $px + \bar py = y$ and hence $x \adh y$. 
	\item Reflexivity is a direct consequence of idempotence. 
		Let $x,y,u,v \in X$ and assume $x \adh y$ and $u \adh v$. Then
		$$
			q(px + \bar p u) + \bar q(py + \bar p v) = p(qx + \bar q y) + \bar p(qu + \bar q v) = py + \bar p v
			.
		$$ 
	\item Direct consequence of reflexivity and convexity of adherence.
	\item Assume $pz  + \bar px \adh pz + \bar py$ and $z$ is $\Alg X$-cancellable. Let $q \in (0,1)$. Then 
		$$pz + \bar py = q(p z + \bar px) + \bar q (pz + \bar py) = pz + \bar p(qx + \bar q y)$$ 
		implies $qx + \bar q y = y$, after cancelling $z$. Hence $x \adh y$.
\qedhere
	\end{enumerate}
\end{proof}

\begin{exa}\label{ex:CAs}~Here are two examples of convex algebras.
	\begin{enumerate}
	\item Let $\Alg V$ be a vector space over $\mathbb R$ and $X \subseteq V$ a convex subset.
		Then $X$ with the operations inherited from $\Alg V$ is a cancellative convex algebra $\Alg X$.
		Conversely, every cancellative convex algebra is isomorphic to a convex subset of a vector space, 
		cf.\ \cite[Theorem~2]{stone:1949} or \cite[Satz~3]{kneser:1952}. 
	\item In particular, we consider the vector space $\ell^1(S)$ for a set $S$. 
		Recall, $\ell^1(S) = \{(r_s)_{s \in S} \mid r_s \in \mathbb{R}, \sum_{s \in S} |r_s| < \infty \}$ 
		with the norm $\|(r_s)_{s \in S}\|_1=\sum_{s \in S} |r_s|$. 
		The set $\Dis S$ of finitely supported 
		probability distributions over $S$ forms a convex subset of $\ell^1(S)$ and hence 
		a cancellative convex algebra $\Alg D_S$. It is shown in \cite[Lemma~1]{neumann:1970} that $\Alg D_S$ 
		is the free convex algebra generated by $S$, i.e., every map of $S$ into some convex algebra $\Alg X$ 
		has a unique extension to a homomorphism from $\Alg D_S$ into $\Alg X$. 
	\end{enumerate}
\end{exa}

The following construction is basic for our considerations.

\begin{defi}\label{def:Pc}
	Let $\Alg X$ be convex algebra. Then $\CSS X$ denotes the set of nonempty convex subsets of $X$, i.e., 
	carriers of subalgebras of $\Alg X$. We endow $\CSS X$ with the pointwise operations 
	\[
		p A + \bar p B = \{p a + \bar p b \mid a \in A, b\in B\}
		.
	\]
\end{defi}

Then $\CSS X$ forms a convex algebra, cf.\ \cite{BSS17}, and we write $\CSS\Alg X$ for this algebra. 
Note that requiring the elements of $\CSS X$ to be nonempty is necessary for the projection axiom to hold. 

We note that $\CSS$ is a monad on $\EM(\Dis)$ as shown 
in~\cite{BSS17}. On morphisms, $\CSS$ acts as the powerset monad.  
The original algebra $\Alg X$ embeds in $\CSS\Alg X$ via the unit of the monad $\eta\colon x\mapsto\{x\}$. 

For convex subsets of a finite dimensional vector space, the pointwise operations are known as 
Minkowski addition and are a basic construction in convex geometry, cf.~\cite{schneider:1993}.

The algebra $\CSS\Alg X$ is in general not cancellative 
and has a nontrivial adherence relation, cf.\ \cite[Example~6.3]{fritz:2015}. However it contains cancellative
elements: It is easy to check that for each $\Alg X$-cancellable element $x$ the element $\{x\}$ is 
$\CSS\Alg X$-cancellable. 

\begin{exa}\label{ex:CAss}~Here are further two examples of convex algebra which are of particular interst 
	in this paper.
	\begin{enumerate}
	\item The motivating example for this work is the convex algebra $\CSS\Alg D_S$ of convex subsets of 
		distributions over a set $S$. 
	\item Let $X$ be the carrier of a meet-semilattice and define $p x + \bar p y = x \wedge y$ for $x, y \in X$ and $p \in (0,1)$. 
		Then $X$ becomes a convex algebra $\Alg X$ with these operations, cf.\ \cite[\S4.5]{neumann:1970}. 
		This algebra is not cancellative, in fact $\adh = \{(x,y) \mid x \ge y\}$. For the categorically minded, we remark
		that behind this construction is the support monad map from $\Dis$ to $\Pow$, the finite powerset monad, and
		semilattices are the Eilenberg-Moore algebras of $\Pow$. 
	\end{enumerate}
\end{exa}

We now present a construction that provides a beautiful way of constructing convex algebras out of existing ones. 

\begin{exa}~\label{ex:semilattice-construction}
	The \emph{semilattice construction}, cf.\ \cite[p.22f]{fritz:2015}: 
	Let $S$ be the carrier of a meet-semilattice and let $(\Alg X_s)_{s \in S}$ be an $S$-indexed family of convex algebras. 
	Moreover, let $(f_s^t)_{\substack{s, t \in S\\ s \le t}}$ be a family of convex algebra homomorphisms 
	$f_s^t \colon \Alg{X}_t \to \Alg{X}_s$ that satisfy $f_s^t \after f_t^u = f_s^u$ for all $s \le t \le u$, 
	and $f_s^s = \id_{X_s}$ for all $s \in S$. 
	Let $X$ be the disjoint union of all $X_s$ for $s \in S$. 
	Then $X$ becomes a convex algebra $\Alg X$ with operations defined by 
	$p x + \bar p y = p f_{s\wedge t}^s(x) + \bar p f_{s\wedge t}^t(y)$ for $x \in X_s$, $y \in X_t$, and $p \in (0,1)$.
	The algebra $\mathbb X$ obtained in this way is the direct limit of the diagram formed by the algebras $\mathbb X_s$ and the maps $f^t_s$.
\end{exa}

\begin{defi}\label{def:ideals}
	Let $\Alg X$ be a convex algebra, and $P, Q \subseteq X$. 
	\begin{itemize}
		\item $P$ is an \emph{ideal} if $\forall x \in P.\,\, \forall y\in X.\,\, \forall p \in (0,1).\,\, px + \bar py \in P$. 
		\item $P$ is a \emph{prime ideal} if it is an ideal and its complement $X\setminus P$ is convex.
		\item $Q$ is an \emph{extremal set} if $px + \bar p y \in Q \Rightarrow x,y \in Q$ for all $x,y \in X, p \in (0,1)$.
		\item $z\in X$ is an \emph{extremal point} if $\{z\}$ is an extremal set. 
			Explicitly: $z$ is an extremal point if whenever $px+\bar py=z$ for $x,y\in X,p\in(0,1)$, 
			it follows that $x=y=z$. 
			The set of all extremal points of $\Alg X$ is denoted as $\Ext\Alg X$.
\qedd
	\end{itemize}
\end{defi} 

Again, the terminology used in the literature is not uniform: in \cite[Definition~7]{jacobs:2010} extremal sets 
are called \emph{filters}. Moreover, let us remark that the construction $\Ext$ is not functorial.

\begin{lem}\label{lem:ideal-extr}
	Let $\Alg X$ be a convex algebra, and $P\subseteq X$. Then $P$ is an ideal if and only if $X \setminus P$ is an extremal set. 
\end{lem}

\begin{proof}
	Assume $P$ is an ideal. Let $x,y \in X, p \in (0,1)$ such that $px + \bar p y \in X \setminus P$. 
	If $x \in P$ or $y \in P$, then since $P$ is an ideal also $px + \bar p y \in P$, a contradiction. 
	Hence $x, y \in X \setminus P$. 
	
	For the converse, assume $X \setminus P$ is extremal and let $x \in P, y \in X, p \in (0,1)$. 
	If $px + \bar p y \notin P$, then $px + \bar p y \in X \setminus P$ which implies $x, y \in X \setminus P$, a contradiction. 
	Hence, $px + \bar p y \in P$.
\end{proof}

\section{The Problem and Some Example Solutions}
\label{sec:problem-ex}

\framebox[\textwidth][l]{
	\raisebox{-4mm}{\rule{0pt}{10mm}}\kern1pt\parbox{\textwidth-5mm}{
	Let $\Alg X$ be a convex algebra. 
	Can one extend it for one element to a convex algebra $\Alg X_*$ with carrier $X \cup \{*\}$ where $* \notin X$? 
	If yes, what are all possible extensions?
	}
}
\\[3mm]
We will show that an arbitrary convex algebra $\Alg X$ can be extended in many ways,
and describe all possible ways of extending $\Alg X=\Alg D_S$ and $\Alg X = \CSS\Alg D_S$. 

First, we provide four examples of extensions, two of which are instances of the semilattice construction 
from Example~\ref{ex:semilattice-construction}.

\begin{exa}~\label{ex:star} 
	Let $\Alg X$ be a convex algebra and $* \notin X$. We denote the operations of $\Alg X$ as before by $p (-) + \bar p (-)$. 
	In each of the examples we construct a convex algebra $\Alg X_*$ with operations denoted by $p (-) \oplus \bar p (-)$ 
	satisfying $px\oplus\bar py=px+\bar py$, $x,y\in X$, $p\in[0,1]$.
	\begin{enumerate}
	\item Black-hole behaviour, cf.\ \cite[Example~6.1]{fritz:2015}: 
		In this example, $*$ behaves like a black hole and swallows everything in the sense that $x \adh *$ for all $x \in X$.
		To be precise, consider the semilattice $S = \{0,1\}$ with $0 \le 1$. Let $\Alg X_0$ be the trivial convex algebra with 
		$X_0 = \{*\}$ and $\Alg X_1 = \Alg X$. Let $f_0^1: \Alg X_1 \to \Alg X_0$ be the unique homomorphism 
		(mapping everything to $*$). Then the semilattice construction gives us a convex algebra $\Alg X_*$ with the property 
		\begin{equation}\label{eq:black-hole}
			px \oplus \bar py = 
			\begin{cases}
				px + \bar p y &\hspace*{-3mm},\quad x, y \in X,\\
				* &\hspace*{-3mm},\quad x = * \text{ or } y = *.
			\end{cases}
		\end{equation}
	\item Imitating behaviour: 
		The intuition behind this construction is that $*$ imitates the behaviour of a given element $w \in X$.
		Formally, we want to build $\Alg X_*$ in such a way that 
		$px\oplus\bar p*=px+\bar p w$ for all $p\in(0,1]$ and $x\in X$. 

		To do this, consider again the semilattice $S = \{0,1\}$ with $0 \le 1$. 
		Let $\Alg X_0 = \Alg X$ and $\Alg X_1$ be the trivial convex algebra with $X_1 = \{*\}$. 
		Let $f_0^1: \Alg X_1 \to \Alg X_0$ be the homomorphism mapping $*$ to $w$. 
		Then the semilattice construction gives us a convex algebra $\Alg X_*$ with the property 
		\begin{equation}\label{eq:imitating}
			px \oplus \bar py =
			\begin{cases}
				px + \bar p y &\hspace*{-3mm},\quad x, y \in X,\\
				px + \bar p w &\hspace*{-3mm},\quad x \in X, y = *,\\ 
				pw + \bar p y &\hspace*{-3mm},\quad x =*, y \in X,\\ 
				* &\hspace*{-3mm},\quad x=y=*.
			\end{cases}
		\end{equation}
	\item Imitating an outer element: 
		Assume we are given a convex algebra $\Alg Y$ which contains $\Alg X$ as a subalgebra. 
		Let $w\in Y\setminus X$ be such that $X\cup\{w\}$ is convex. 
		Then we obtain an extension $\Alg X_*$ by identifying $X\cup\{*\}$ with $X\cup\{w\}$ 
		via $x\mapsto x$ for $x\in X$ and $*\mapsto w$. We say that $*$ imitates the 
		outer element $w$, since $px\oplus\bar p*=px+\bar p w$ for all $p\in(0,1]$ and $x\in X$. 

		This way of defining extensions is of course trivial, 
		but it is useful in presence of a natural larger algebra. 
		For example, we will apply it when $D$ is a convex subset of a vector 
		space $\Alg V$, $\Alg X = \CSS\Alg D$, and $\Alg Y = \CSS\Alg V$. 
	\item Mixed behaviour: Let $w$ be an extremal point of $\Alg X$. 
		The intuition in this example is that $*$ imitates $w \in X$ on $X \setminus \{w\}$ and swallows $w$. 
		That is, we want to build $\Alg X_*$ according to 
		\begin{equation}\label{eq:mixed}
			px \oplus \bar py = 
			\begin{cases}
				px + \bar p y &\hspace*{-3mm},\quad x, y \in X,\\
				px + \bar p w &\hspace*{-3mm},\quad x \in X \setminus\{w\}, y = *,\\ 
				pw + \bar p y &\hspace*{-3mm},\quad x =*, y \in X\setminus\{w\},\\ 
				* &\hspace*{-3mm},\quad \text{otherwise}.
			\end{cases}
		\end{equation}
		The fact that this construction indeed produces a convex algebra is not an instance of the 
		semilattice construction and requires a proof, which 
		we give in Section~\ref{sec:free} below (p.\pageref{prf:ex:star.4}).
	\end{enumerate}
\end{exa}

\section{Extensions of Convex Algebras - The Prime Ideal}
\label{sec:necessity}
 
The following two notions provide a crucial characteristic of an extension $\Alg X_*$ for a convex algebra $\Alg X$. 

\begin{defi}\label{def:Ad-Pri}
	Let $\Alg X$ be a convex algebra, and let $\Alg X_*$ be an extension. 
	Then its \emph{set of adherence} $\Adh(\Alg X_*)$ is
	$\Adh(\Alg X_*) =\{x\in X\mid x \adh *\}$
	and its \emph{prime ideal} is  $\Pri(\Alg X_*) = X\setminus\Adh(\Alg X_*)$. 
\hspace*{0pt}\qedd
\end{defi}

\begin{lem}\label{lem:Ad-Pri}
	Let $\Alg X$ be a convex algebra, and let $\Alg X_*$ be an extension of $\Alg X$. 
	The set $\Pri(\Alg X_*)$ is indeed a prime ideal of $\Alg X$. 
\end{lem}
\begin{proof}
	Let $x\in\Pri(\Alg X_*)$, $y \in X$, $p\in(0,1)$. Then
	$$q(px + \bar py) + \bar q * = \overline{q\bar p}\left(\frac{qp}{\overline{q\bar p}}x + \frac{\bar q}{\overline{q\bar p}}*\right) + q\bar p y \in X$$
	since $y \in X$ and $\frac{qp}{\overline{q\bar p}}x + \frac{\bar q}{\overline{q\bar p}}* \in X$ due to $x \in \Pri(\Alg X_*)$ and hence $x \notin \Adh(\Alg X_*)$. 
	Therefore, $px + \bar py \in \Pri(\Alg X_*)$ proving that $\Pri(\Alg X_*)$ is an ideal in $\Alg X$. 
	By Lemma~\ref{lem:adh-props}.2 $\Adh(\Alg X_*)$ is convex and hence $\Pri(\Alg X_*)$ is prime.
\end{proof}

The next lemma gives a way to conclude that $*$ imitates an element. 

\begin{lem}\label{lem:uniqueness-lemma}
	Let $\Alg Y$ be a convex algebra, $\Alg X\leq\Alg Y$ a subalgebra, 
	and let $\Alg X_*$ be an extension of $\Alg X$. 
	Further, let $z\in\Pri(\Alg X_*)$ and assume that $z$ is $\Alg Y$-cancellable. 
	If there exist $w\in\Alg Y$ and $q\in(0,1)$ with $qz+\bar q*=qz+\bar qw$, then 
	$*$ imitates $w$ on $\Pri(\Alg X_*)$ and $\Adh(\Alg X_*)\subseteq\{x\in X\mid x\adh w\}$. 
\end{lem}
\begin{proof}
	Let $x\in\Pri(\Alg X_*)$, $p\in(0,1)$, and set $s=\frac{\bar p}{\overline{pq}}$. Then $s\in(0,1)$ and 
	$\overline{sq}\cdot p=\bar s$, $\overline{sq}\cdot\bar p=s\bar q$, and we have 
	\[
		sqz+\overline{sq}(\underbrace{px+\bar p*}_{\in\Pri(\Alg X_*)\subseteq Y})=
		s(qz+\bar q*)+\bar sx=s(qz+\bar qw)+\bar sx=
		sqz+\overline{sq}(\,\underbrace{px+\bar pw}_{\in Y}\,)
		.
	\]
	Cancelling $z$ yields $px+\bar p*=px+\bar pw$. We conclude that indeed $*$ imitates $w$ on all of 
	$\Pri(\Alg X_*)$. 
	Assume now that $x\in\Adh(\Alg X_*)$. Then by Lemma~\ref{lem:adh-props}.3.
	\[
		pz+\bar px\adh pz+\bar p*=pz+\bar pw,\quad \text{~for~} p\in(0,1)
		.
	\]
	Again using cancellability of $z$, it follows that $x\adh w$ by Lemma~\ref{lem:adh-props}.4. 
\end{proof}

\section{Extensions of Free Algebras and Functoriality}
\label{sec:free}

Let $S$ be a nonempty set and consider the free convex algebra over $S$. 
As noted in Example~\ref{ex:CAs}.2, this is the algebra $\Alg D_S$ of finitely supported distributions on $S$. 
In the next theorem we determine all possible one-point extensions of $\Alg D_S$. 

\begin{thm}\label{thm:free-extensions}
	Let $S$ be a nonempty set and consider the free convex algebra $\Alg D_S$.
	One-point extensions of $\Alg D_S$ can be constructed as follows:
	\begin{enumerate}
	\item The black-hole behaviour, where the set of adherence equals $\Dis S$. 
	\item Let $w\in\Dis S$, and let $*$ imitate $w$ on all of $\Dis S$. 
	\item Let $w$ be an extremal point of $\Alg D_S$, and let $*$ imitate $w$ on $\Dis S\setminus\{w\}$ and adhere $w$. 
	\end{enumerate}
	Every one-point extension of $\Alg D_S$ can be obtained in this way, and each two of these extensions are different. 
\end{thm}

Note that $w\in\Ext\Alg D_S$ if and only if $w$ is a corner point, in other words, a Dirac measure concentrated 
at one of the points of $S$. 

The fact that the constructions $(1)$ and $(2)$ give extensions is Example~\ref{ex:star}.1/2. 
The construction in $(3)$ is Example~\ref{ex:star}.4, for which we will now provide evidence. 
First, we prove a more general statement that we call the gluing lemma, which will be needed later as well. It gives a way to produce extensions with a prescribed set of adherence.

\begin{lem}[Gluing Lemma]\label{lem:glue}
	Let $\Alg X$ be a convex algebra, and $P \subseteq X$ a prime ideal. 
	Assume we have convex operations $p(-) \boxplus \bar p(-)$ on $\Alg P_*$ that extend $\Alg P$ (whose operations are
	inherited from $\Alg X$). Assume further that $\Adh(\Alg P_*) = \emptyset$ and that
	\begin{equation}\label{eq:glue-lem}
		px + \bar py \adh p x \boxplus \bar p*, \quad \text{~for~} x \in P, \,y \in X\setminus P, \,p \in (0,1).
	\end{equation}
	Then the operations $p(-) \oplus \bar p(-)$, $p\in(0,1)$, defined as follows extend $\Alg X$ to a 
	convex algebra $\Alg X_*$ with $\Adh(\Alg X_*) = X \setminus P$:
	$$px \oplus \bar py = \left\{\begin{array}{ll}
		px + \bar p y, &  x, y \in X,\\	
		px \boxplus py, & x=*,y\in P\text{ or }x\in P,y=*,\\
		*, & \text{otherwise}.
	\end{array}
	 \right.$$
\end{lem}

\begin{proof}
	We first show two auxilliary properties for $x \in P$ and $y \in X\setminus P$: (\ref{eq:A}) $\Leftrightarrow$ (\ref{eq:B}) and (\ref{eq:A}) $\Leftrightarrow$ (\ref{eq:C}) for 
	\begin{align}
		& \forall p \in (0,1). \,\,px + \bar p y \adh px \boxplus \bar p*
		\label{eq:A}
		\\
		& \forall q,r \in (0,1). \,\, q(rx \boxplus \bar r *) + \bar q y = qrx \boxplus \overline{qr}* 
		\label{eq:B}\\
		& \forall q,r \in (0,1). \,\, q(rx + \bar r y) \boxplus \bar q * = qrx \boxplus \overline{qr}*
		\label{eq:C} 
	\end{align}
Let $x \in P, y \in X \setminus P$ and $s,t \in (0,1)$. 
For (\ref{eq:A}) $\Leftrightarrow$ (\ref{eq:B}) we first compute
\begin{eqnarray*}
	s(tx + \bar t y) + \bar s(t x  \boxplus \bar t *) & \stackrel{(b)}{=} &
	stx + s \bar t y + \bar s(tx \boxplus \bar t *)
	\\ 
	& \stackrel{(b)}{=} &
	\overline{s\bar t} \left( \frac{st}{\overline{s\bar t}}x + \frac{\bar s}{\overline{s\bar t}}(tx  \boxplus \bar t *)\right) + s\bar t y
	\\ 
	& \stackrel{(c)}{=} &
	\overline{s\bar t} \left( \frac{st}{\overline{s\bar t}}x \boxplus \frac{\bar s}{\overline{s\bar t}}(tx  \boxplus \bar t *)\right) + s\bar t y
	\\ 
	& \stackrel{(d)}{=} &
	\overline{s\bar t} \left( \frac{t}{\overline{s\bar t}}x \boxplus \frac{\bar s \bar t}{\overline{s\bar t}}*\right) + s\bar t y
\end{eqnarray*}
and refer to this equality as $(\dagger)$.
Here, $(b)$ is an application of barycenter in $\Alg X$; $(c)$ holds
since $x, tx \boxplus \bar t * \in P$; and $(d)$ is an application of barycenter in $\Alg P_*$.    

Now assume (\ref{eq:A}), let $q,r \in (0,1)$, and take $s = \frac{\bar q}{\overline{qr}}$ and $t = qr$. Then $s, t \in (0,1)$ and
\[
	q(rx \boxplus \bar r *) + \bar q y = 
	\overline{s\bar t} \left( \frac{t}{\overline{s\bar t}}x \boxplus \frac{\bar s \bar t}{\overline{s\bar t}}*\right) + s\bar t y 
	\stackrel{(\dagger)}{=}
	s(tx + \bar t y) + \bar s(t x  \boxplus \bar t *)
	\stackrel{(\ref{eq:A})}{=} 
	t x  \boxplus \bar t * = qrx \boxplus \overline{qr} *
\]
proving (\ref{eq:B}).

For (\ref{eq:B}) $\Rightarrow$ (\ref{eq:A}), assume (\ref{eq:B}), let $p \in (0,1)$, and take any $q, r\in (0,1)$ such that 
$p = qr$. Set again $s = \frac{\bar q}{\overline{qr}}$. Then $s\in (0,1)$ and 
$$px \boxplus \bar p* = qrx \boxplus \overline{qr} * \stackrel{(\ref{eq:B})}{=} 
q(rx \boxplus \bar r*) + \bar q y \stackrel{(\dagger)}{=} s(px + \bar py)+  \bar s(px  \boxplus \bar p*)$$
which proves (\ref{eq:A}).

For (\ref{eq:A}) $\Leftrightarrow$ (\ref{eq:C}) we now compute
\begin{eqnarray*}
	s(tx + \bar t y) \boxplus \bar s(t x  \boxplus \bar t *) & \stackrel{(d)}{=} & 
	s(tx + \bar t y) \boxplus \bar st x  \boxplus \bar s\bar t *\\ & \stackrel{(d)}{=} &
	\overline{\bar s\bar t} \left( \frac{s}{\overline{\bar s\bar t}}(tx + \bar ty) \boxplus \frac{\bar st}{\overline{\bar s\bar t}}x\right) \boxplus \bar s\bar t *\\ & \stackrel{(b)}{=} &
	\overline{\bar s\bar t} \left( \frac{t}{\overline{\bar s\bar t}}x + \frac{s \bar t}{\overline{\bar s\bar t}}y\right) \boxplus \bar s\bar t *
\end{eqnarray*}
and refer to this equality as $(\ddagger)$.

Now assume (\ref{eq:A}), let $q,r \in (0,1)$ be given, consider $p = qr$, and take $s = \frac{q \bar r}{\overline{qr}}$ and $t = p$. Hence $\bar s = \frac{\bar q}{\overline{qr}}$ and $t = qr$. Then $s, t \in (0,1)$, $\frac{t}{\overline{\bar s\bar t}} = r$, $\overline{\bar s \bar t} = q$ and hence
$$q(rx + \bar r y) \boxplus \bar q* \stackrel{(\ddagger)}{=} \frac{q\bar r}{\overline{qr}}(qrx + \overline{qr}y) \boxplus \frac{\bar q}{\overline{qr}}(qrx \boxplus \overline{qr}*) \stackrel{(\ref{eq:A})}{=} qr x \boxplus \overline{qr}*.$$

For the opposite direction, assume (\ref{eq:C}). Let $p$ be given. Pick $q,r \in (0,1)$ such that $p = qr$ and 
put $s = \frac{q\bar r}{\overline{qr}}$. Then again
$$px + \bar p* = qr x + \overline{qr}* \stackrel{(\ref{eq:C})}{=} q(rx + \bar ry) \boxplus \bar q* \stackrel{(\ddagger)}{=} s(px + \bar p y) \boxplus \bar s(px 
\boxplus \bar p*)$$ proving (\ref{eq:A}).

Coming back to the actual proof of the lemma, note that (\ref{eq:A}) is satisfied by assumption, hence both (\ref{eq:B}) and
(\ref{eq:C}) hold as well. Moreover, recall that $X\setminus P$ is a subalgebra of $\Alg X$, since $P$ is a prime ideal. 
Write $P' = X \setminus P$ and $\Alg P'$ for the corresponding convex algebra. 
Further, let $\Alg P'_*$ be the black-hole extension of $\Alg P'$ (operations denoted as $p(-)+'\bar p(-)$). 

We will show now that $p(-) \oplus \bar p(-)$ are indeed convex operations on $\Alg X_*$, i.e., 
we check idempotence, parametric commutativity, and parametric associativity, cf.\ Proposition~\ref{prop:bin-suff}. 

Since the operations of $\Alg X$ and $\Alg P_*$ coincide on $P$ and those of $\Alg X$ and $\Alg P'_*$ coincide 
on $P'$, the definition of $p(-) \oplus \bar p(-)$ gives
\[
	px \oplus \bar py = 
	\begin{cases}
		px + \bar py &\hspace*{-3mm},\quad x, y \in X,\\
		px \boxplus \bar py &\hspace*{-3mm},\quad x, y \in P_*,\\
		px+'\bar py &\hspace*{-3mm},\quad x, y \in P'_*.
	\end{cases}
\]
We have $X_*=P\cup P'\cup\{*\}$ and, by what we just observed, 
the operations $p(-)\oplus\bar p(-)$ coincide on the union of each two of these 
sets with the operations of a convex algebra (namely with $+$ of $\Alg X$ on $P\cup P'$, with $\boxplus$ of $\Alg P_*$ 
on $P\cup\{*\}$, and with $+'$ of $\Alg P'_*$ on $P'\cup\{*\}$). 
The idempotence law involves only one variable and parametric commutativity involves only two variables. 
We conclude that both of these laws hold for $\oplus$. Parametric associativity is the law
\[
	p(qx \oplus \bar q y) \oplus \bar pz = 
	pqx \oplus \overline{pq}(\frac{p\bar q}{\overline{pq}}y \oplus \frac{\bar p}{\overline{pq}}z)
	.
\]
If all three of $x,y,z$ belong to one of $P\cup P'$, $P\cup\{*\}$, and $P'\cup\{*\}$, 
the above argument shows that this law holds. 
It remains to check six cases. In the following set $r = \frac{p\bar q}{\overline{pq}}$. 
\begin{enumerate}
\item $x \in P, y \in P', z = *$. 
	\begin{eqnarray*}
		p(qx \oplus \bar q y) \oplus \bar pz & = & p(qx + \bar q y)\boxplus \bar p*\\
		& \stackrel{(\ref{eq:C})}{=} & pqx \boxplus \overline{pq}*\\
		& = & pqx \boxplus \overline{pq}(ry \oplus \bar rz)\\
		& = & pqx \oplus \overline{pq}(ry \oplus \bar rz). 		
	\end{eqnarray*}
\item $x \in P, y = * , z \in P'$.
	\begin{eqnarray*}
		p(qx \oplus \bar q y) \oplus \bar pz & = &  p(qx \oplus \bar q *) \oplus \bar pz\\
		& = & 	p(qx \boxplus \bar q *) \oplus \bar pz\\
		& = & 	p(qx \boxplus \bar q *) + \bar pz\\
		& \stackrel{(\ref{eq:B})}{=} & pq x \boxplus \overline{pq} *\\
		& = & pqx \boxplus \overline{pq}(r* \oplus \bar rz)\\
		& = & pqx \oplus \overline{pq}(ry \oplus \bar rz).
	\end{eqnarray*}
\item $x \in P', y \in P , z = *$. 
	\begin{eqnarray*}
		p(qx \oplus \bar q y) \oplus \bar pz & = & p(qx \oplus \bar q y) \oplus \bar p*\\
		& \stackrel{comm}{=} & p(\bar qy \oplus q x) \oplus \bar p *\\
		& = & p(\bar qy + q x) \boxplus \bar p *\\
		&\stackrel{(\ref{eq:C})}{=} & p\bar q y \boxplus \overline{p\bar q}*\\
		& = & p\bar q y \oplus \overline{p\bar q}*
	\end{eqnarray*} and
	\begin{eqnarray*}
	    pqx \oplus \overline{pq}(ry \oplus \bar rz) & \stackrel{comm}{=} & \overline{pq}(ry \oplus \bar rz) \oplus pqx\\
	    & = & \overline{pq}(ry \boxplus \bar r*) + pqx\\
	    & \stackrel{(\ref{eq:B})}{=} & \overline{pq}ry \boxplus \overline{\overline{pq}r}*\\
	    & = & \overline{pq}ry \oplus \overline{\overline{pq}r}*\\
	    & = & p \overline q y \oplus \overline{p\overline q} *.
	\end{eqnarray*}
\item $x \in P', y = *, z \in P$.
	\begin{eqnarray*}
		p(qx \oplus \bar q y) \oplus \bar pz & = & p* \oplus \bar p z\\
		&\stackrel{comm}{=} & \bar pz \oplus p *
	\end{eqnarray*} 
	and 
	\begin{eqnarray*}
	    pqx \oplus \overline{pq}(ry \oplus \bar rz) & \stackrel{comm}{=} & \overline{pq}(\bar rz \oplus r*) \oplus pqx\\
	    & = & \overline{pq}(\bar rz \boxplus r*) + pqx\\
	    &\stackrel{(\ref{eq:B})}{=} & \overline{pq}\bar rz \boxplus \overline{\overline{pq}\bar r}*\\
	    & = & \bar p z \oplus p *.
	\end{eqnarray*} 
\item $x = *, y \in P, z \in P'$.
	\begin{eqnarray*}
	   p(qx \oplus \bar q y) \oplus \bar pz  & \stackrel{comm}{=} & p(\bar qy \oplus q *) \oplus \bar pz\\
	   & = & p(\bar qy \boxplus q *) + \bar pz\\
	   & \stackrel{(\ref{eq:B})}{=} & p\bar q y \boxplus \overline{p \bar q}*\\
	   & = & p\bar q y \oplus \overline{p \bar q}*
	\end{eqnarray*} and
	\begin{eqnarray*}
	    pqx \oplus \overline{pq}(ry \oplus \bar rz) & \stackrel{comm}{=} & \overline{pq}(ry \oplus \bar rz) \oplus pq*\\
	    & = & \overline{pq}(ry + \bar rz) \oplus pq*\\
	    & \stackrel{(!)}{=} & \overline{pq}(ry + \bar rz) \boxplus pq*\\
	    & \stackrel{(\ref{eq:C})}{=} & \overline{pq}ry \oplus \overline{\overline{pq}r}*\\
	    & = & p \bar q y \oplus \overline{p \bar q} *,
	\end{eqnarray*}
	where the equality marked by $(!)$ holds because $ry + \bar r z \in P$.
\item $x = *, y \in P', z \in P$.
	\begin{eqnarray*}
	   p(qx \oplus \bar q y) \oplus \bar pz  & = & p* \oplus \,\bar p z\\
	  & \stackrel{comm}{=} & \bar p z \oplus p *
	\end{eqnarray*} and
	\begin{eqnarray*}
	   pqx \oplus \overline{pq}(ry \oplus \bar rz) & \stackrel{comm}{=} & \overline{pq}(\bar r z \oplus ry) \oplus pq*\\
		   & = & \overline{pq}(\bar r z + ry) \oplus pq*\\
		   & \stackrel{(!!)}{=} & \overline{pq}(\bar r z + ry) \boxplus pq*\\
		   & \stackrel{(\ref{eq:C})}{=} & \overline{pq}\bar r z \boxplus \overline{\overline{pq}\bar r}*\\
		   & = & \bar p z \oplus p *,
     \end{eqnarray*}
     where the equality marked by $(!!)$ holds because now $\bar rz +  r y \in P$.       
\qedhere
\end{enumerate}
\end{proof}

\begin{proof}[Proof of Example~\ref{ex:star}.4]
\label{prf:ex:star.4}
	Assume we are in the situation of Example~\ref{ex:star}.4, i.e., $\Alg X$ is a convex algebra and 
	$w$ is an extremal point of $\Alg X$. Set $P=X\setminus\{w\}$, then $P$ is a prime ideal. Further, 
	let $\Alg P_*$ be obtained as in Example~\ref{ex:star}.3 with $\Alg P\leq\Alg X$ by letting $*$ imitate $w$. 
	Condition (\ref{eq:glue-lem}) is satisfied with equality, and hence the Gluing Lemma provides 
	$\Alg X_*$. The operations $p(-)\oplus\bar p(-)$ obtained in this way coincide with those written in 
	Example~\ref{ex:star}.4. 
\end{proof}

\begin{proof}[Proof of Theorem~\ref{thm:free-extensions}]
	The uniqueness part is easy to see. First, the action of $*$ determines which case of $(1)$--$(3)$ occurs since 
	$\Adh((\Alg D_S)_*)$ is $\Dis S$ in case $(1)$, $\emptyset$ in case $(2)$, and $\{w\}$ in case $(3)$.
	Now uniqueness of the point $w$ in $(2)$ and $(3)$ follows since $\Alg D_S$ is cancellative. 

	We have to show that every extension occurs in one of the described ways. Hence, let an extension 
	$(\Alg D_S)_*$ be given.
	If $\Pri((\Alg D_S)_*)=\emptyset$, case $(1)$ takes place. Assume that $\Pri((\Alg D_S)_*)\neq\emptyset$ 
	and choose $z\in\Pri((\Alg D_S)_*)$ and $q\in(0,1)$. Set 
	\[
		w=\frac 1{\bar q}\big([qz+\bar q*]-qz\big)\in\ell^1(S)
		,
	\]
	then $qz+\bar q*=qz+\bar qw$ by definition. 
	We apply Lemma~\ref{lem:uniqueness-lemma} with $\Alg D_S\leq\ell^1(S)$ and $z,w,q$. This yields 
	\begin{equation}\label{eq:imitates}
		px+\bar p*=px+\bar pw,\quad x\in\Pri((\Alg D_S)_*),p\in(0,1)
		,
	\end{equation}
	and $\Adh((\Alg D_S)_*)\subseteq\{x\in\Dis S \mid x\adh w\}\subseteq\{w\}$. 

	As a linear combination of two elements of $\Dis S$, the element $w$ is finitely supported. Further, 
	by (\ref{eq:imitates}), 
	\[
		1=\frac 1{\bar p}\big(\|pz+\bar p*\|_1-p\|z\|_1\big)\leq\|w\|_1\leq
		\frac 1{\bar p}\big(\|pz+\bar p*\|_1+p\|z\|_1\big)=\frac{1+p}{1-p}
	\]
	for all $p\in(0,1)$, and we see that $\|w\|_1=1$. Together, $w\in\Dis S$. 

	If $\Pri((\Alg D_S)_*)=\Dis S$, we are in case $(2)$ of the theorem. Otherwise, 
	$\Pri((\Alg D_S)_*)=(\Dis S)\setminus\{w\}$. This implies that $w$ is an extremal point of $\Alg D_S$, 
	and we are in case $(3)$. 
\end{proof}

\sloppypar Next we investigate functoriality of one-point extensions. We say that a functor $F\colon\EM(\Dis)\to\EM(\Dis)$ 
naturally provides a one-point extension, if 
$\Alg X\leq F\Alg X$ and $F\Alg X$ has carrier $X \cup \{*\}$ for $* \notin X$ for every algebra $\Alg X$, and 
$(Ff)|_{X}=f$ for every convex map $f\colon\Alg X\to\Alg Y$. The latter property is (literally) a natural property: it 
says that the family of inclusion maps $\iota_X\colon\Alg X\to F\Alg X$ is a natural transformation of the identity functor to $F$. 

An example of a functor possessing these properties is obtained by the black-hole construction:
for an algebra $\Alg X$ let $F\Alg X$ be its black-hole extension, and for a convex map $f:\Alg X\to\Alg Y$ 
let $Ff$ be the extension of $f$ mapping $*$ (of $F\Alg X$) to $*$ (of $F\Alg Y$). 

\begin{thm}\label{thm:functor}
	Let $F:\EM(\Dis)\to\EM(\Dis)$ be a functor such that 
	for all objects $\Alg X$ and for all morphisms $f:\Alg X\to\Alg Y$
	\begin{equation}\label{eq:reasonable}
		\Alg X\leq F\Alg X,\ \text{the carrier of } F\Alg X \text{ is }X \cup \{*\}\text{ with }*\not\in X,\qquad 
		\raisebox{5mm}{$\xymatrix@C=20mm@R=5mm{
			F\Alg X \ar[r]^{Ff} & F\Alg Y
			\\
			\Alg X \ar[r]_{f} \ar[u]^{\iota_X} & \Alg Y \ar[u]_{\iota_Y}
		}$}
	\end{equation}
	Then, for all $\Alg X$, $F\Alg X$ is the black-hole extension, and for all $f\colon\Alg X\to\Alg Y$, $Ff$ 
	is the extension of $f$ mapping $*$ (of $F\Alg X$) to $*$ (of $F\Alg Y$).
\end{thm}

We present the proof using two lemmata. 

\begin{lem}\label{lem:1st}
	Assume that $F\colon\EM(\Dis)\to\EM(\Dis)$ satisfies (\ref{eq:reasonable}), and let $f\colon\Alg X\to\Alg Y$ be a 
	convex map. Then $(Ff)(*)=*$ and $f(\Pri(F\Alg X))\subseteq\Pri(F\Alg Y)$, 
	$f(\Adh(F\Alg X))\subseteq\Adh(F\Alg Y)$.
\end{lem}
\begin{proof}
	For the proof of $(Ff)(*)=*$, note that $(Ff)^{-1}(\{*\})\subseteq\{*\}$ since 
	$(Ff)|_X=f$. If $f$ has a right inverse, say $g:\Alg Y\to\Alg X$ with $f\circ g=\id_{\Alg Y}$, then 
	$(Ff)((Fg)(*))=*$, and hence $(Fg)(*)=*$. In turn also $(Ff)(*)=*$. 
	Now let $f$ be arbitrary. Let $\Alg Z$ be an algebra which has only one element, a final object 
	of $\EM(\Dis)$, and let $h:\Alg Y\to\Alg Z$ 
	be the unique convex map. The map $h\circ f$ has a right inverse, and therefore $(Fh)((Ff)(*))=(F(h\circ f))(*)=*$. 
	Again, we obtain $(Ff)(*)=*$. 

	It remains to prove that $f$ maps the respective prime ideals (sets of adherence) into each other. 
	Let $x\in X$ and $p\in(0,1)$. Then 
	\begin{equation}\label{eq:F-computation}
		pf(x)+\bar p*=p(Ff)(x)+\bar p(Ff)(*)=(Ff)(px+\bar p*)=
		\begin{cases}
			f(px+\bar p*)\in Y &\hspace*{-3mm},\ x\in\Pri(F\Alg X),
			\\[1mm]
			(Ff)(*)=* &\hspace*{-3mm},\ x\in\Adh(F\Alg X).
		\end{cases}
	\end{equation}
	Thus, indeed, $f(x)\in\Pri(F\Alg Y)$ if $x\in\Pri(F\Alg X)$, and $f(x)\in\Adh(F\Alg Y)$ if $x\in\Adh(F\Alg X)$.
\end{proof}

\begin{lem}\label{lem:2nd}
	Assume that $F\colon\EM(\Dis)\to\EM(\Dis)$ satisfies (\ref{eq:reasonable}), and let $S$ be an infinite set. 
	Then $F\Alg D_S$ is the black-hole extension of $\Alg D_S$. 
\end{lem}
\begin{proof}
	Assume towards a contradiction that $\Pri(F\Alg D_S)\neq\emptyset$. By Theorem~\ref{thm:free-extensions} we find 
	$w\in\Dis S$ such that $px+\bar p*=px+\bar pw$, $x\in\Pri(F\Alg D_S)$, $p\in(0,1)$. Fix $x\in\Pri(F\Alg D_S)$ and $p\in(0,1)$, 
	and let $f\colon\Alg D_S\to\Alg D_S$ be an automorphism. Then $f(x)\in\Pri(F\Alg D_S)$ by Lemma~\ref{lem:1st}, and we 
	can compute 
	\[
		pf(x)+\bar pw=pf(x)+\bar p*\stackrel{(\ref{eq:F-computation})}{=} f(px+\bar p*)=
		f(px+\bar pw)=pf(x)+\bar pf(w)
		.
	\]
	Cancelling $f(x)$ gives $w=f(w)$. Hence $w$ is a fixpoint of every automorphism.

	Since $S$ is infinite, we can choose a point $s_1\in S$ which lies outside of the support of $w$. 
	Further, let $s_2\in S$ be in the support of $w$, and let $\sigma\colon S\to S$ be the permutation of $S$ which 
	exchanges $s_1$ and $s_2$ and leaves all other points fixed. Since $\Alg D_S$ is free with basis $S$, this permutation 
	extends to an automorphism $f$ of $\Alg D_S$. But now $f(w)\neq w$, a contradiction. 
\end{proof}

\begin{proof}[Proof of Theorem~\ref{thm:functor}]
	The fact that $Ff$ is the extension of $f$ mapping $*$ to $*$ was shown in Lemma~\ref{lem:1st}. It remains to 
	show that, for every algebra $\Alg X$, $F\Alg X$ is the black-hole extension. Given $\Alg X$,  
	choose an infinite set $S$ and a surjective convex map $f:\Alg D_S\to\Alg X$. This is possible since every convex algebra is the image of a free convex algebra, and if $S \supseteq S'$ then there is a surjective homomorphism from  $\Alg D_S$ to $\Alg D_{S'}$. Then, by  Lemma~\ref{lem:1st} and Lemma~\ref{lem:2nd}, 
	$\Adh(F\Alg X)\supseteq f(\Adh(F\Alg D_S)) = f(\Alg D_S)=\Alg X$.
\end{proof}

\section{Extensions of $\CSS\Alg D$}
\label{sec:PcD}

In this section we formulate and prove Theorem~\ref{thm:Pc-extensions} 
where we describe the set of all extensions $(\CSS\Alg D)_*$ for convex algebras $\Alg D$ 
which are convex subsets of a vector space (equivalently, cancellative) and satisfy a certain linear boundedness condition. 
Theorem~\ref{thm:Pc-extensions} applies in particular to the algebra $\Alg D=\Alg D_S$ of finitely supported distributions over $S$. 

We start with some algebraic preliminaries. First, we recall the notion of linear boundedness, see e.g. 
\cite[Definition~1.1]{boerger.kemper:1996}.

\begin{defi}\label{def:lin-bound}
	A convex algebra $\Alg X$ is \emph{linearly bounded}, if every homomorphism of the convex algebra $(0,\infty)$ 
	into $\Alg X$ is constant. 
\qedd
\end{defi}

Intuitively, a convex algebra is linearly bounded if it does not contain an infinite ray.
A large class of examples of linearly bounded algebras is given by topologically bounded subsets of a topological vector space. 
Recall that a topological vector space is a vector space endowed with a topology such that addition and scalar 
multiplication are continuous. 
Our standard reference for the theory of topological vector spaces is \cite{rudin:1991}. 

\begin{defi}\label{def:bound-tvs}
	Let $\Alg V$ be a topological vector space. A subset $D\subseteq V$ is \emph{bounded}, if 
	for every neighbourhood $U$ of $0$ there exists $r_0>0$ such that $D\subseteq rU$, $r>r_0$
	(cf.\ \cite[p.8]{rudin:1991}).
\end{defi}

For example, if $\Alg V$ is a normed space (with a norm denoted by $\|.\|$), then a subset $D$ is bounded in this sense
if and only if $\sup_{x\in D}\|x\|<\infty$. 

We could not find an explicit reference for the following (intuitive) fact, and hence provide a complete proof.

\begin{lem}\label{lem:bound-lin-bound}
	Let $\Alg V$ be a topological vector space. Then for every bounded convex subset $D$ of $\Alg V$, the convex 
	algebra $\Alg D$ is linearly bounded. 
\end{lem}
\begin{proof}
	Let $D$ be a bounded convex and nonempty subset of a topological vector space $\Alg V$, and let 
	$\varphi\colon(0,\infty)\to\Alg D$ be a convex map. 
	By \cite[Proposition~2.7]{flood:1981} we find a convex extension $\Phi:\mathbb R\to\Alg V$ of $\varphi$. 
	Set $\Psi=\Phi-\Phi(0)$, then $\Psi$ is convex and $\Psi(0)=0$. The purpose of this normalisation is that it 
	allows us to conclude $\Psi(tx)=t\Psi(x)$, $t>0$, $x\in\mathbb R$: If $t=1$ this is trivial. If $t<1$ use convexity 
	to compute $\Psi(tx)=\Psi(tx+(1-t)0)=t\Psi(x)+(1-t)\Psi(0)=t\Psi(x)$. If $t>1$, use the already known to compute 
	$\Psi(x)=\Psi(\frac 1t\cdot tx)=\frac 1t\Psi(tx)$. 

	Assume now towards a contradiction that $\varphi$ is not constant. Then we can choose $s\in(0,\infty)$ with 
	$\varphi(s)\neq\Phi(0)$, i.e., $\Psi(s)\neq 0$. Choose a neighbourhood $U$ of $0$ such that $\Psi(s)\notin U$. 
	Since $D$ is bounded, also its translate $D-\Phi(0)$ is bounded. Hence, we find $r>0$ with $D-\Phi(0)\subseteq rU$. 
	From 
	\[
		r\Psi(s)=\Psi(rs)=\varphi(rs)-\Phi(0)\in D-\Phi(0)\subseteq rU
	\]
	we obtain $\Psi(s)\in U$, and reached a contradiction. 
\end{proof}

\begin{rem}\label{rem:lin-bound-vs}
	Let $\Alg V$ be a vector space over $\mathbb R$. Then, for each fixed $w\in\Alg V$ and $t\in\mathbb R\setminus\{0\}$, 
	we have the translation map $x\mapsto x+w$ and the scaling map $x\mapsto tx$. They are bijective convex maps on $\Alg V$. Applying $\CSS$ on these maps gives bijective convex maps on $\CSS\Alg V$. 
	Moreover, a subset $A\in\CSS\Alg V$ is linearly bounded if and only if $t(A+w)$ is linearly bounded. 
\end{rem}

The following observation holds for all cancellative convex algebras $\Alg D$.

\begin{lem}\label{lem:adh-singletons}
	Let $\Alg D$ be a convex algebra and consider $\Alg X =\CSS\Alg D$. 
	If $\Alg D$ is cancellative, then $A\adh\{x\}\Rightarrow A=\{x\}$ for all $A\in X$, $x\in D$.
\end{lem}
\begin{proof}
	Let $a \in A$. Then $pa + \bar p x = x = px + \bar p x$ which after cancelling with $x$ yields $a = x$. 
	Since $A$ is nonempty, as it belongs to $\CSS\Alg D$, we get $A = \{x\}$. 
\end{proof}

Under a linear boundedness condition, the roles of $A$ and $\{x\}$ can be exchanged.

\begin{lem}\label{lem:intersection}
	Let $\Alg V$ be a vector space over $\mathbb R$, let $A\in\CSS\Alg V$, and assume that $A-A$ is linearly bounded. 
	Then 
	\[
		\bigcap_{p\in[0,1)}\big(p\{x\}+\bar pA\big) \subseteq \{x\}
		,\quad\text{~for~} x\in V
		.
	\]
	In particular, $\{x\}\adh A\Rightarrow A=\{x\}$ for all $x\in V$. 
\end{lem}
\begin{proof}
	Note first that $A - A$ is convex. Let $y$ belong to the intersection. 
	Then $y\in A$ and for each $p\in(0,1)$ we find $a_p\in A$ with 
	$y=px+\bar pa_p$. This implies 
	\[
		\frac p{\bar p}(x-y)=y-a_p\in A-A,\quad\text{for~} p\in(0,1)
		.
	\]
	Any positive real number $t$  can be represented as $\frac p{\bar p}$, namely with $p = \frac{t}{1+t}\in(0,1)$. 
	It is easy to check then that $\varphi\colon t\mapsto t(x-y)$ is a convex homomorphism from $(0,\infty)$ to $A-A$. 
	Since $A-A$ is linearly bounded, $\varphi$ is constant, which further implies $x=y$. 
\end{proof}

In order to construct extensions where $*$ imitates an outer element, we need the following notion of visibility closure. 

\begin{defi}\label{def:vis}
	Let $\Alg X$ be a convex algebra and $A\in\CSS\Alg X$. The \emph{visibility hull} of $A$ is 
	\[
		\ViH(A)=\big\{x\in X\mid \forall a\in A.\,\,\forall p\in(0,1).\,\, px+\bar pa\in A\big\}
		.
	\]
	The set $A$ is \emph{visibility closed} if $A=\ViH(A)$. 
\qedd
\end{defi}

\begin{exa}\label{ex:vis-example}
	Let $A\subseteq\mathbb R^2$ be the open half-disk 
	$A=\{(t_1,t_2)\in\mathbb R^2\mid t_1^2+t_2^2<1,t_2>0\}$. 
	Then $\ViH(A)$ is the closed half disk, shown in Figure~\ref{fig:Vis-A}.

	Now consider $B=A\cup\{(0,0)\}$. Then the part of the boundary of $B$ located on the $t_1$-axis 
	does not belong to $\ViH(B)$, see Figure~\ref{fig:Vis-B}.
	\begin{figure}[h]
	    \centering
	    \begin{subfigure}[h]{0.4\textwidth} \centering
		\includegraphics[width=0.5\textwidth]{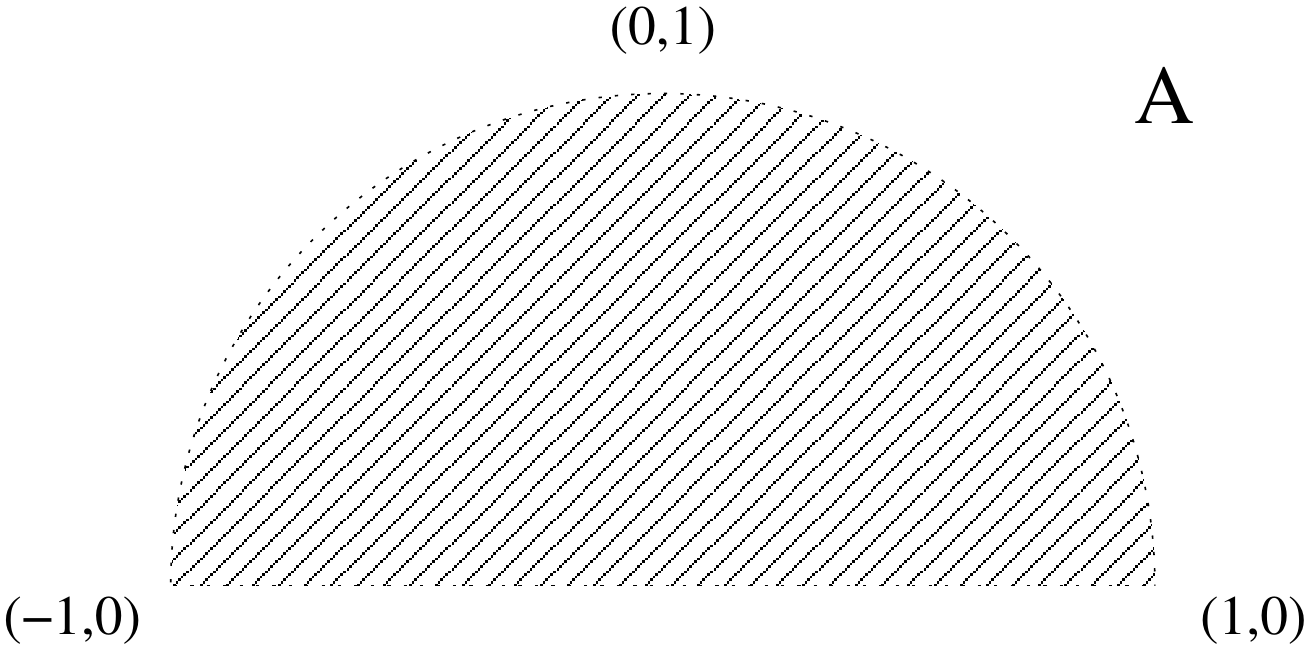}
		\includegraphics[width=0.5\textwidth]{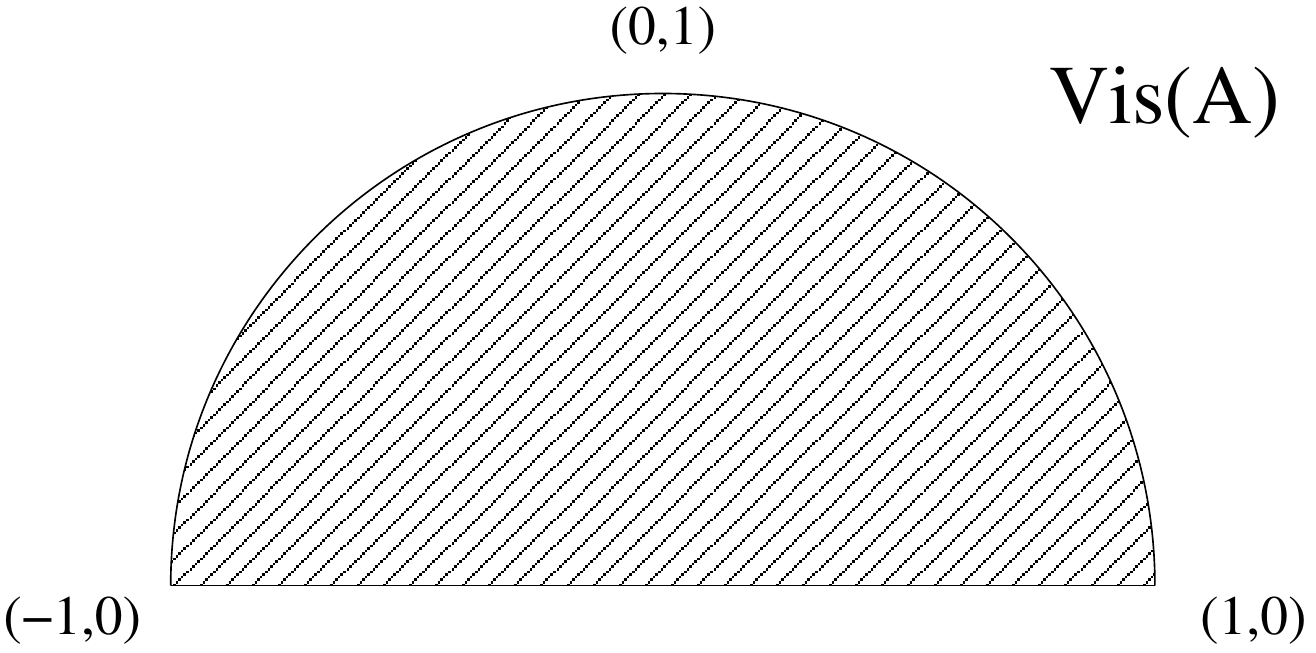}
		\caption{$A$ and $\ViH(A)$\label{fig:Vis-A}}
	    \end{subfigure}
	    ~ 
	    \begin{subfigure}[h]{0.4\textwidth} \centering
		\includegraphics[width=0.5\textwidth]{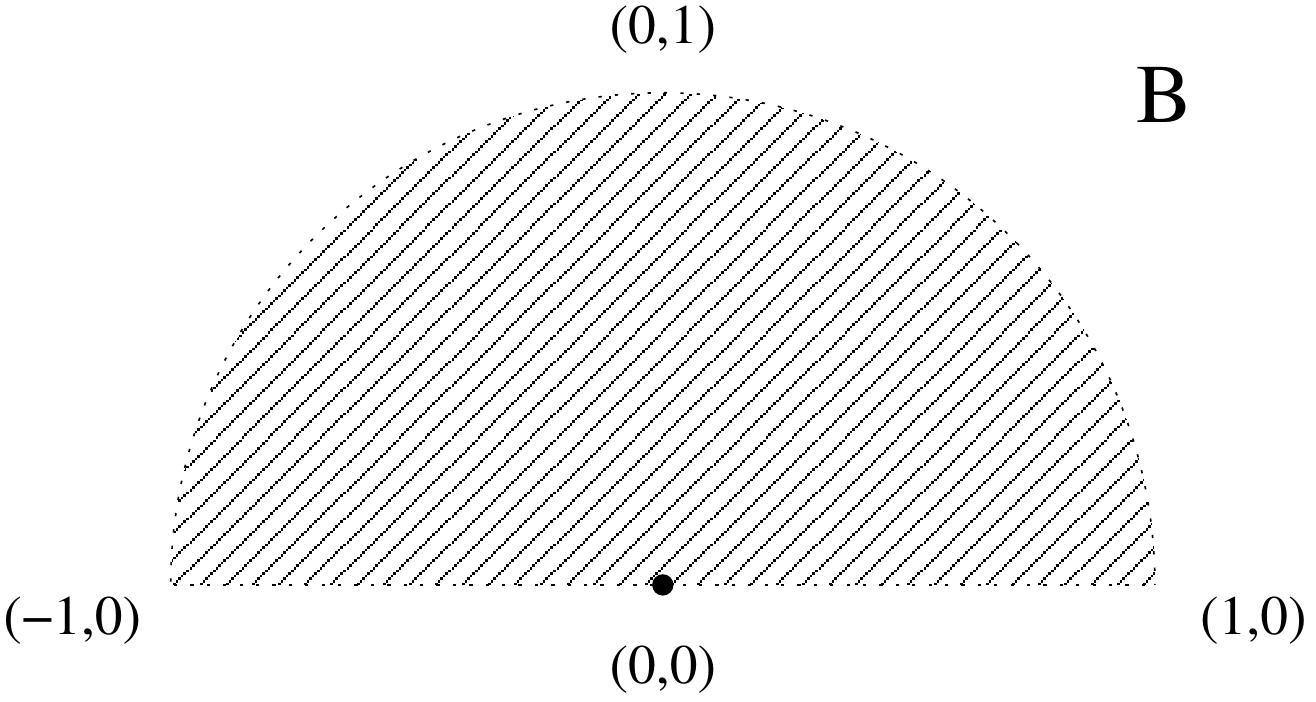}
		\includegraphics[width=0.5\textwidth]{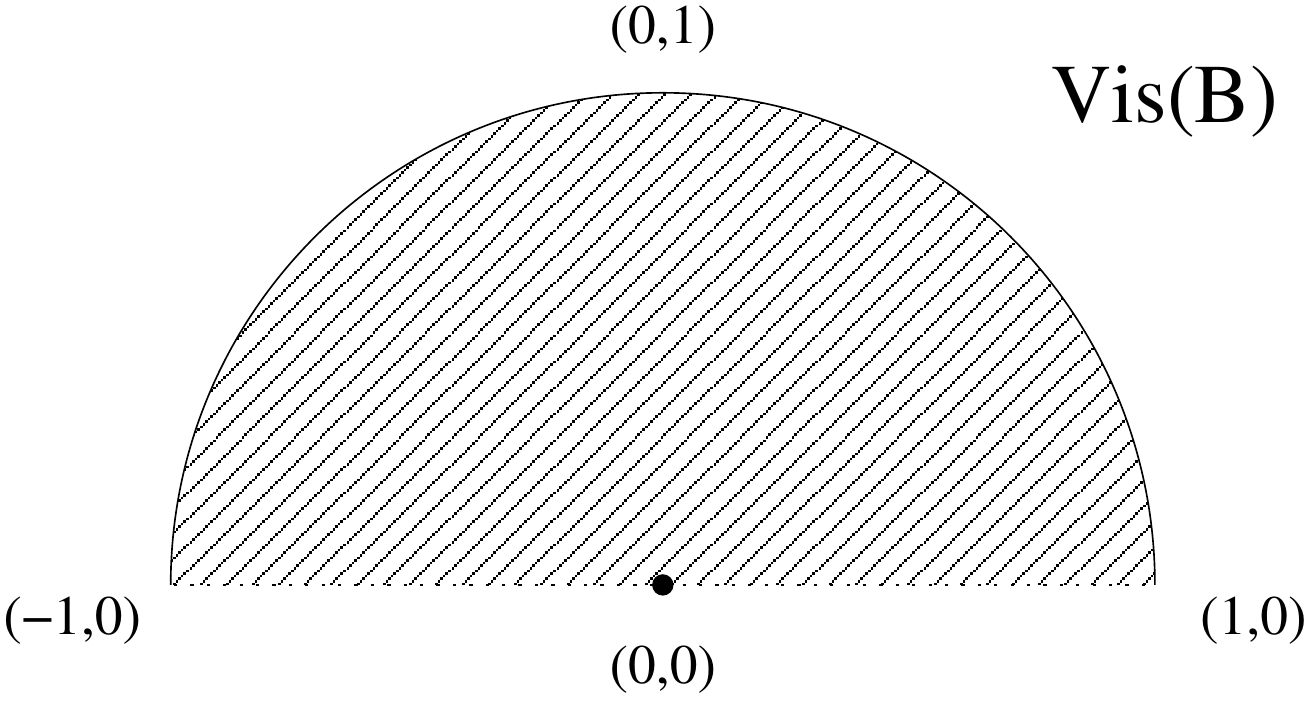}
		\caption{$B$ and $\ViH(B)$\label{fig:Vis-B}}	
	    \end{subfigure}
	    \caption{Visibility hulls}
	\end{figure}
\end{exa}

Let $\Alg V$ be a vector space over $\mathbb R$. 
The \emph{affine hull} of a subset $A\subseteq V$ is 
$$\AfH(A)=\{\sum_{i=1}^nt_ix_i\mid n\geq 1,x_i\in A,t_i\in\mathbb R,\sum_{i=1}^nt_i=1\}.$$
The affine hull of $A$ is the smallest affine subspace of $\Alg V$ containing $A$, see e.g.\ \cite[p.6]{rockafellar:1970}.

\begin{lem}\label{lem:vis-properties}
	Let $\Alg V$ be a vector space over $\mathbb R$, and $A\in\CSS\Alg V$. Then 
	\begin{enumerate}
	\item ${\displaystyle \ViH(A)=\bigcap_{\substack{a\in A\\ p\in(0,1)}}\frac 1p(A-\bar pa)\subseteq\AfH(A)}$.
	\item $\ViH(A)$ is convex.
	\item $A\subseteq\ViH(A)$ and $\ViH(\ViH(A))=\ViH(A)$.
	\item $\ViH(\{z\})=\{z\}$ for all $z\in V$. 
	\item If $\Alg V$ is a topological vector space, then $\ViH(A)\subseteq\overline A$, $\overline A$ being the topological closure of $A$. 
	\end{enumerate}
\end{lem}
\begin{proof}
\hfill
	\begin{enumerate}
	\item We have
		\[
			x\in\ViH(A) \ \Leftrightarrow\ \forall a\in A.\,\,\forall p\in(0,1).\,\, px+\bar pa\in A
			\ \Leftrightarrow\ \forall a\in A.\,\,\forall p\in(0,1).\,\,x\in\frac 1p(A-\bar pa)
		\]
	\item
	By 1., the set $\ViH(A)$ is the intersection of convex sets. 
	\item Let $x\in A$. Then $px+\bar p a\in A$, $a\in A$, $p\in(0,1)$, since $A$ is convex. Thus $A\subseteq\ViH(A)$. 
		Assume that $x\in\ViH(\ViH(A))$, and let $a\in A$, $p,q\in(0,1)$. 
		Then $px+\bar pa\in\ViH(A)$, since $a \in A \subseteq \ViH(A)$, and hence 
		$qp x+\overline{qp}a=q(px+\bar pa)+\bar qa\in A$.
		Every number $r\in(0,1)$ can be represented as $r=pq$ with some $p,q\in(0,1)$, and we conclude that 
		$x\in\ViH(A)$. 
	\item We have $\frac 1p(\{z\}-\bar pz)=\{z\}$, $p\in(0,1)$. By 1., $\ViH(\{z\})=\{z\}$. 
	\item Let $x\in\ViH(A)$ and $a\in A$. Then $x=\lim_{p\rightarrow 1}(px+\bar pa)\in\overline{A}$. 
\qedhere
	\end{enumerate}
\end{proof}

The operator $\ViH\colon\CSS\Alg V\to\CSS\Alg V$ is not monotone, as demonstrated in Example~\ref{ex:vis-example}. 
Hence, it is not the restriction of a topological closure operator to $\CSS\Alg V$. 
Still, it is related to topological closures:

\begin{rem}\label{rem:relclos}
	Let $\Alg V$ be a topological vector space and $A\in\CSS\Alg V$ \emph{relatively closed}, i.e., 
	closed in $\AfH(A)$ w.r.t.\ the subspace topology. Then $A$ is visibility closed. 
	This follows by putting together Lemma~\ref{lem:vis-properties}.1 and 5. 
	The converse does not hold, as demonstrated by the set $\ViH(B)$ from Example~\ref{ex:vis-example}. 
	This observation shows for example that $\ViH(\Dis S)=\Dis S$. 
\end{rem}

We can now formulate our description of extensions of $\CSS\Alg D$. 

\begin{thm}\label{thm:Pc-extensions}
	Let $\Alg V$ be a vector space over $\mathbb R$, let $D$ be a convex subset of $V$ with more than one element, 
	and consider the convex algebra $\Alg X=\CSS\Alg D$.
	One-point extensions of $\Alg X$ can be constructed as follows:
	\begin{enumerate}
	\item The black-hole behaviour, where the set of adherence equals $X$. 
	\item Let $C\in\CSS(\ViH(D))$, 
		and let $*$ imitate $C$ on all of $X$. 
	\item Let $w$ be an extremal point of $\Alg D$, 
		and let $*$ imitate $\{w\}$ on $X\setminus\{\{w\}\}$ and adhere $\{w\}$. 
	\item Let $C\in\CSS(\ViH(D))$ with at least two elements, 
		assume $\CoH\{A\in X\mid A\nadh C\}\neq X$, and let $I = \CoH\{A\in X\mid A\nadh C\}$. 
		Let $P \neq X$ be a prime ideal in $\Alg X$ with $I\subseteq P$, 
		and let $*$ imitate $C$ on $P$ and adhere $X\setminus P$. 
	\end{enumerate}
	Assume in addition that $D-D$ is linearly bounded. 
	Then every one-point extension of $\Alg X$ can be obtained in this way.
	Each two of these extensions are different: 
	the point $w$ in case $(3)$, the set $C$ in cases $(2),(4)$, and the prime ideal $P$ in case $(4)$, 
	are uniquely determined by a given extension. 
\end{thm}

We are familiar with the constructions $(1)$--$(3)$ from Example~\ref{ex:star} and Theorem~\ref{thm:free-extensions}. 
That $(4)$ gives extensions follows from the Gluing Lemma, Lemma~\ref{lem:glue}. 

\begin{proof}[Proof of Theorem~\ref{thm:Pc-extensions}; constructions]
	The black-hole behaviour is always possible, cf.\ Example~\ref{ex:star}.1.
	Assume we are given $C \in \CSS(\ViH(D))$. If $C\subseteq D$ use Example~\ref{ex:star}.2.
	Otherwise, use Example~\ref{ex:star}.3 with the algebra $\Alg Y=\CSS\Alg V$ and its element $C$. 
	The necessary hypothesis, that $X\cup\{C\}$ is convex, is satisfied since $C\subseteq\ViH(D)$ and hence 
	$pA+\bar pC\subseteq D$ for all $A\subseteq D$, $p\in(0,1)$. 
	Assume we are given an extremal point $w$ of $\Alg D$. 
	Then $\{w\}$ is an extremal point of $\CSS\Alg D$. The construction $(3)$ is exactly 
	Example~\ref{ex:star}.4 applied with this extremal point.

	Now consider the construction $(4)$. Assume $C$ and $P$ have the properties stated in $(4)$. 
	We first show 
	\begin{equation}\label{eq:in-P}
		pA+\bar pC\in P,\quad\text{~for~} A\in P, \, p\in(0,1)
		.
	\end{equation}
	Assume towards a contradiction that $pA+\bar pC\not\in P$ for some $A\in P$, $p\in(0,1)$. 
	Since $P\supseteq I$ we have $pA+\bar pC\notin I$ and hence $pA+\bar pC\adh C$. 
	Since $C\subseteq\ViH(D)$ we have $sA+\bar sC\in X$, for $s\in(0,1)$. 
	Choose $q\in(0,1)$, then 
	\[
		C=q(pA+\bar pC)+\bar qC=qpA+\overline{qp}C\in X
		.
	\]
	Since $P$ is an ideal in $\Alg X$, we get $pA+\bar pC\in P$, a contradiction. 

	The relation (\ref{eq:in-P}) implies that $P\cup\{C\}$ is a convex subset of $\CSS\Alg V$, and 
	Examples~\ref{ex:star}.2/3 provide an extension $\Alg P_*$. 
	Again by (\ref{eq:in-P}), it holds that $\Adh(\Alg P_*)=\emptyset$. 
	To apply the Gluing Lemma, we need to check (\ref{eq:glue-lem}). 
	Let $A\in P$, $B\in X\setminus P$, $p\in(0,1)$. Since $I\subseteq P$ we have $B\adh C$, and hence 
	$pA+\bar pB\adh pA+\bar pC=pA+\bar p*$ by Lemma~\ref{lem:adh-props}.3.
\end{proof}

Assume that $D-D$ is linearly bounded. 
Our task is to show that every given extension $\Alg X_*$ can be realised as described in $(1)$--$(4)$ of the theorem,
and show uniqueness. 
The proof relies on the following lemma.

\begin{lem}\label{lem:one-singleton}
	Assume $\Alg X_*$ is an extension with $\Pri(\Alg X_*)\neq\emptyset$. 
	Then $\Adh(\Alg X_*)$ contains at most one singleton set. 
\end{lem}
\begin{proof}
	Assume that $\{x\},\{y\}\in\Adh(\Alg X_*)$, and choose $A\in\Pri(\Alg X_*)$. 
	Then, for each $p\in(0,1)$, by Lemma~\ref{lem:adh-props}.3
	\[
		pA+\bar p\{x\}\adh pA+\bar p*,\quad pA+\bar p\{y\}\adh pA+\bar p*
		.
	\]
	Set $C=pA+\bar p*$. Then $C\in\Pri(\Alg X_*)$ and for each $q\in(0,1)$ 
	\[
		q(pA+\bar p\{x\})+\bar qC=C=q(pA+\bar p\{y\})+\bar qC
		.
	\]
	Thus, for each $a\in A,c\in C$ we find $a_1\in A,c_1\in C$ with 
	\[
		q\bar px+\overline{q\bar p}\Big(\frac{qp}{\overline{q\bar p}}a+\frac{\bar q}{\overline{q\bar p}}c\Big)
		=q\bar py+\overline{q\bar p}\Big(\frac{qp}{\overline{q\bar p}}a_1+\frac{\bar q}{\overline{q\bar p}}c_1\Big)
		,
	\]
	and hence 
	\[
		\frac{q\bar p}{\overline{q\bar p}}(x-y)\in D-D
		.
	\]
	Any positive real number $t$ can be represented as $\frac{q\bar p}{\overline{q\bar p}}$ with some $p,q\in(0,1)$, 
	for example use $p=\frac{1}{2t+1}, q=\frac{2t+1}{2t+2}$. 
	Thus $\varphi\colon t\mapsto t(x-y)$ is a homomorphism of $(0,\infty)$ to $D-D$. 
	Since $D-D$ is linearly bounded, $\varphi$ is constant, and hence $x =y$. 
\end{proof}

\begin{proof}[Proof of Theorem~\ref{thm:Pc-extensions}; all $\Alg X_*$ are obtained.]
	Let an extension $\Alg X_*$ of $\Alg X$ be given. If $\Pri(\Alg X_*)=\emptyset$ then case $(1)$ 
	of the theorem holds. Assume in the following that $\Pri(\Alg X_*)\neq\emptyset$. 

	By Lemma~\ref{lem:one-singleton}, $\Adh(\Alg X_*)$ contains at most one singleton set. 
	Since $D$ has more than one element, we find $z\in D$ with $\{z\}\in\Pri(\Alg X_*)$. 
	Choose $q\in(0,1)$. Then $q\{z\}+\bar q*\in\Pri(\Alg X_*)\subseteq\CSS\Alg V$. 
	We will show that $*$ imitates the convex set 
	\[
		C = \frac 1{\bar q}\big([q\{z\}+\bar q*]-qz\big)\in\CSS\Alg V
		.
	\]
	By definition, $C$ satisfies $q\{z\}+\bar q*=q\{z\}+\bar qC$. 
	Since singletons are $\CSS\Alg V$-cancellable, as noted after Definition~\ref{def:Pc}, 
	the hypothesis of Lemma~\ref{lem:uniqueness-lemma} are fulfilled. 
	We conclude that $*$ imitates $C$ on $\Pri(\Alg X_*)$ and that 
	$\Adh(\Alg X_*)\subseteq\big\{A\in X\mid A\adh C\big\}$. 

	Consider the case that $\Adh(\Alg X_*)$ contains a singleton, say $\{w\}\in\Adh(\Alg X_*)$. 
	Since $C\subseteq\frac 1{\bar q}(D-qz)$, the set $C-C$ is linearly bounded, cf.\ Remark~\ref{rem:lin-bound-vs}. 
	Lemma~\ref{lem:intersection} implies that $C=\{w\}$ and Lemma~\ref{lem:adh-singletons} that
	$\{A\in X\mid A\adh C\}=\{\{w\}\}$. We see that $\Pri(\Alg X_*)=X\setminus\{\{w\}\}$ 
	and that $*$ imitates $\{w\}$ on $\Pri(\Alg X_*)$. 
	Since $X\setminus\{\{w\}\}$ is an ideal in $\Alg X$, also $D\setminus\{w\}$ is an ideal in $\Alg D$, i.e., 
	$w$ is an extremal point of $\Alg D$. Thus $\Alg X_*$ has the form described in case $(3)$. 

	Consider the case that $\Adh(\Alg X_*)$ contains no singleton. Hence all singletons are in $\Pri(\Alg X_*)$. Then 
	\[
		p\{y\}+\bar pC=p\{y\}+\bar p*\subseteq X
		,\quad\text{~for~} y\in D,\,\, p\in(0,1)
		. 
	\]
	Thus $C\subseteq\bigcap_{\substack{y\in D\\ p\in(0,1)}}\frac{1}{\bar p}(D-py)=\ViH(D)$, by
	Lemma~\ref{lem:vis-properties}.1. 
	If $\Adh(\Alg X_*)=\emptyset$, case $(2)$ of the theorem holds. 
	Assume that $\Adh(\Alg X_*)\neq\emptyset$. If $C$ contains only one element, say $C=\{w\}$, 
	we would have 
	\[
		\emptyset\neq\Adh(\Alg X_*)\subseteq\{A\in X\mid A\adh \{w\}\}=\{\{w\}\}
		.
	\]
	From this $\Adh(\Alg X_*)=\{\{w\}\}$, a contradiction. Thus $C$ has at least two elements. 
	Since 
	\[
		X\neq\Pri(\Alg X_*)=X\setminus\Adh(\Alg X_*)\supseteq\{A\in X\mid A\nadh C\}
		,
	\]
	the convex hull of $\{A\in X\mid A\nadh C\}$ is not $X$ and case $(4)$ of the theorem holds. 
\end{proof}

\begin{proof}[Proof of Theorem~\ref{thm:Pc-extensions}; uniqueness]
	The uniqueness assertion of the theorem follows since $\Pri(\Alg X_*)$ always contains singletons
	by Lemma~\ref{lem:one-singleton}, and singletons are cancellable in $\CSS\Alg V$. 
\end{proof}

The following example shows that the linear boundedness condition in Theorem~\ref{thm:Pc-extensions} 
cannot be dropped without admitting other types of constructions. 

\begin{exa}\label{ex:lb-nec}
	Let $\Alg V=\mathbb R^2$ and $D=\{(t_1,t_2)\in\mathbb R^2\mid t_2>0\}\cup\{(0,0)\}$. 
	Set $P=\{A\in\CSS\Alg D\mid (0,0)\not\in A\}$. 
	First we show that $P$ is a prime ideal of $\CSS\Alg D$. Denote by $f\colon\mathbb R^2\to\mathbb R$ the projection onto 
	the second coordinate. Let $A\in P$, $B\in\CSS\Alg D$, and $p\in(0,1)$. Then, for each $a\in A$ and $b\in B$, 
	\[
		f(pa+\bar pb)=p\underbrace{f(a)}_{>0}+\bar p\underbrace{f(b)}_{\geq 0}>0
		.
	\]
	Thus $pA+\bar pB\in P$, and we see that $P$ is an ideal. If $A,B\in(\CSS\Alg D)\setminus P$, the point 
	$(0,0)$ belongs to both of $A$ and $B$ and hence also to each convex combination of $A$ and $B$. 
	Thus $(\CSS\Alg D)\setminus P$ is convex, and $P$ is prime. 

	Set $C=D\cup\{(t_1,t_2)\in\mathbb R^2\mid t_2=0,t_1>0\}$. Then $C\subseteq\ViH(D\setminus\{(0,0)\})$, and we 
	can define an extension $\Alg P_*$ by letting $*$ imitate $C$ on all of $P$. 
	We check that the compatibility condition (\ref{eq:glue-lem}) of the gluing lemma is satisfied. 
	To this end we show that every element of $(\CSS\Alg D)\setminus P$ adheres to $C$ (in $\CSS\Alg V$), and 
	refer to Lemma~\ref{lem:adh-props}.3. Let $B\in(\CSS\Alg D)\setminus P$ and $p\in(0,1)$. 
	Since $B\subseteq D\subseteq C$ and $C$ is convex, we have $pB+\bar pC\subseteq C$. For the reverse inclusion, 
	observe that $tC=C$ for all $t>0$. We can thus write any element $c\in C$ as 
	\[
		c=p\cdot(0,0)+\bar p\cdot\big(\frac 1{\bar p}c\big)\in pB+\bar pC
		.
	\]
	Applying the Gluing Lemma we obtain an extension $(\CSS\Alg D)_*$. 
	This extension is not among the ones listed in Theorem~\ref{thm:Pc-extensions}, since $C\nsubseteq\ViH(D)$. 
\end{exa}

Unboundedness of $D$ enters in this example in the way that it enables us to let $*$ imitate a cone. 
In fact, dropping linear boundedness, one can still show that an extension
which is not of type $(1)$--$(4)$ of the theorem, must be such that 
$*$ imitates some cone whose apex lies in $D$. However, we have no description of which cones occur that way.

\section{Termination in Probabilistic Automata}
\label{sec:PA}

A (simple) \emph{probabilistic automaton} (PA, for short)~\cite{SL94,Seg95:thesis} is a triple 
$M = (S, A, \to)$ where $S$ is the set of states, $A$ the set of actions, and $\to \,\, \subseteq S \times A \times \Dis S$ the transition relation. As usual, we write $a \stackrel{a}{\to} \xi$ for $(s,a,\xi) \in\, \to$ and say that $s$ makes an $a$-step to $\xi$.

A probabilistic automaton $M = (S, A, \to)$ can be identified~\cite{BSV04:tcs,Sokolova11} with a $(\Po\Dis)^A$-coalgebra $(S,c_M)$ where $c_M \colon S \to (\Po\Dis S)^A$ on $\Sets$ and $s \stackrel{a}{\to} \xi$ in $M$ iff $\xi \in c_M(s)(a)$.

We say that a PA $M$ is \emph{input enabled}, if for all $s \in S$ and all $a \in A$, the set $\{\xi \in \Dis S \mid s \stackrel{s}{\to} \xi\}$ is nonempty. Input-enabledness implies that the associated coalgebra $c_M$ is actually a $(\Po_{\neq\emptyset}\Dis)^A$-coalgebra where $\Po_{\neq\emptyset}$ denotes the nonempty powerset functor.

We call this view on PA (their standard definition and the observation that they are coalgebras on $\Sets$) the classical view.

In the classical view, the canonical semantics for PA is bisimilarity. However, it is possible to give PA more intricate and useful \emph{distribution} semantics. This was noticed in the last decade by many authors, and studied from a coalgebraic perspective in~\cite{BSS17}, by viewing them as \emph{belief-state transformers} and making the underlaying convex algebra structure explicit. To be precise, we state the following property, which is an instance of~\cite[Lemma 25]{BSS17}. By $U$ we denote the forgetful functor from $\EM(\Dis)$ to $\Sets$.

\begin{lem}\label{lem:BST}
There is a 1-1 correspondence between input enabled PA, i.e.,
$(\Po_{\neq\emptyset}\Dis)^A$-coalgebras on $\Sets$, and $\CSS^A$-coalgebras on
$\EM(\Dis)$ with carriers free algebras: 

$$
\begin{array}{c}
{c_M\colon S \to (\Po_{\neq\emptyset}\Dis S)^A \,\,\,\text{~in~}\,\,\, \Sets}
\raisebox{-6pt}{\rule{0pt}{0pt}}
\\
\hhline{=}
{c_M^\#\colon \Alg D_S \to (\CSS \Alg D_S)^A\,\,\,\text{~in~}\,\,\, \EM(\Dis)}	
\raisebox{13pt}{\rule{0pt}{0pt}}
\end{array}
$$

Given $c_M$, $c_M^\# = \alpha \after \Dis c$ for $\alpha$ being the convex algebra structure on $(\CSS\Alg D_S)^A$. Concretely, for $\xi = \sum p_i s_i \in \Alg D_S$ and $a \in A$ we have
$$c_M^\#(\xi)(a) = \left\{ \,\,\sum p_i \xi_i \,\mid \, s_i \stackrel{a}{\to} \xi_i\right\}.$$

Given $c_M^\#$, $c_M = U c_M^\# \after \eta$ where $\eta$ is the unit of the distribution monad, i.e., $\eta(s) = \delta_s$ is the Dirac distribution that assigns probability 1 to the state $s$. Concretely, for $s \in S$ and $a \in A$ we have
$c_M(s)(a) = c_M^\#(\delta_s)$. \qed
\end{lem}

The fact that input is always enabled is critical here, as $\CSS$ on $\EM(\Dis)$ is the \emph{nonempty} convex powerset and, as discussed above after Definition~\ref{def:Pc}, this nonemptiness is crucial in order to get a convex algebra structure.

Our original motivation when starting the work was to answer the question: What are all possible one-point extensions of $\CSS \Alg D_S$ in order to be able to allow for termination in a belief-state transformer, i.e., in order to overcome the restriction of input-enabledness. 

Theorem~\ref{thm:Pc-extensions} answers this question fully: All one-point extensions of $\Alg X = \CSS \Alg D_S$ with $X = \CSS\Dis S$ are given by (1)-(4) below and they are all different. 
\begin{enumerate}
	\item  The black-hole extension, where the set of adherence equals $X$.
	\item Let $C \in X$ and let $*$ imitate $C$ on all of $X$.
	\item Let $w = \delta_s$ be a Dirac distribution for $s \in S$, and let $*$ imitate $\{w\}$ on $X \setminus \{\{w\}\}$ and adhere $\{w\}$. 
	\item Let $C\in X$ with at least two elements, 
		assume $\CoH\{A\in X\mid A\nadh C\}\neq X$, and let $I = \CoH\{A\in X\mid A\nadh C\}$. 
		Let $P \neq X$ be a prime ideal in $\Alg X$ with $I\subseteq P$, 
		and let $*$ imitate $C$ on $P$ and adhere $X\setminus P$. 
	\end{enumerate}

Here, we get a simplification of the formulation of Theorem~\ref{thm:Pc-extensions} as (as noted above Theorem~\ref{thm:Pc-extensions}) $\ViH(\Alg D_S) = \Alg D_S$ and the extremal points of $\Alg D_S$ are the Dirac distributions on $S$.
The case (4) does not simplify further, which we elaborate on some examples in Section~\ref{sec:appendix-ex}.

Moreover, by Theorem~\ref{thm:functor}, we know that only the black-hole extension is functorial, and 
a 1-1 correspondence between (not necessarily input-enabled) PA,
$(\Po\Dis)^A$-coalgebras on $\Sets$, and $(\CSS + 1)^A$-coalgebras on
$\EM(\Dis)$ with carriers free algebras, where $- + 1$ is the black-hole
extension functor on $\EM(\Dis)$, is also an instance of~\cite[Lemma 25]{BSS17}.

\section{Illustrative Examples of Type (4) Extensions}
\label{sec:appendix-ex}

In what follows, in particular in Example~\ref{ex:conv-body} and Example~\ref{ex:dis-S-finite}, 
we discuss the possible extensions of $\CSS\Alg D$ when $D$ is a compact convex subset of $\mathbb R^n$ 
(which is not a singleton), and in particular when $D=\Dis S$ when $S$ is a finite set. 
The proofs of the stated facts depend heavily on arguments from convex geometry and topology and are deferred 
to Section~\ref{sec:appendix-examples}. 
Here we only recall the necessary terminology and a few facts 
taken from the textbooks \cite{rockafellar:1970} and \cite{schneider:1993}. 

\begin{defi}\label{def:imported-notions-topology}
	Let $A\in\CSS\mathbb R^n$. The \emph{relative interior} $\relint(A)$ of $A$ is the interior of $A$ 
	as a subset of $\AfH(A)$, where $\AfH(A)$ is endowed with the subspace topology inherited from $\mathbb R^n$. 
	The set $A$ is \emph{relatively open}, if $A=\relint(A)$
	(cf.\ \cite[p.44]{rockafellar:1970}).
\qedd
\end{defi}

\begin{defi}\label{def:imported-notions-geometry}
	\hspace*{0pt}
	\begin{enumerate}
	\item A \emph{convex body} is a compact, convex, and nonempty subset of an euclidean space $\mathbb R^n$. 
		The set of all convex bodies in $\mathbb R^n$ is denoted as $\mathcal K^n$
		(cf.\ \cite[p.8]{schneider:1993}).
	\item For given $D\in\mathcal K^n$, by $\mathcal K(D)$ we denote the set $\mathcal K(D)=\{A\in\mathcal K^n\mid A\subseteq D\}$
		(cf.\ \cite[p.157]{schneider:1993}).
	\item Two convex bodies $A,B\in\mathbb R^n$ are \emph{homothetic}, if either one of them is a singleton or 
		there exist $s>0,x\in\mathbb R^n$ with $B=sA+x$
		(\cite[p.xii]{schneider:1993}).
	\item A convex body $A\in\mathcal K^n$ is \emph{indecomposable}, if a representation $A=B+C$ is only 
		possible with $B$ and $C$ both homothetic to $A$
		(cf.\ \cite[p.150]{schneider:1993}).
\qedd
	\end{enumerate}
\end{defi}

\begin{rem}\label{rem:imported-facts-geometry}
	\hspace*{0pt}
	\begin{enumerate}
	\item When endowed with the Minkowski sum and (pointwise declared) multiplication with positive 
		scalars, $\mathcal K^n$ is a convex cone, in particular a convex algebra. 
		These operations are continuous w.r.t.\ the Hausdorff metric. 
		See \cite[p.42,\,p.126]{schneider:1993}. Recall in this place that the Hausdorff metric $d_H$ is defined 
		as 
		\[
			d_H(A,B)=\max\big\{
			\sup_{a\in A}\inf_{b\in B}d(a,b),\sup_{b\in B}\inf_{a\in A}d(a,b)
			\big\}
		\]
		where $d$ denotes the euclidean metric of $\mathbb R^n$. 
	\item The set $\mathcal K(D)$ carries a subalgebra of $\CSS\Alg D$. It is compact w.r.t. the Hausdorff metric. 
		See \cite[Theorem~1.8.6]{schneider:1993}.
	\item On the set of all convex bodies having more than one point, homothety induces an equivalence relation. 
	\end{enumerate}
\end{rem}

\begin{exa}\label{ex:conv-body}
	Let $n\geq 1$, let $D$ be a compact convex subset of $\mathbb R^n$ with more than one point, and consider 
	$\Alg X =\CSS\Alg D$. 
	The first observation is that $\ViH(D)=D$ since $D$ is closed. Hence, the extensions $\Alg X_*$ 
	described in $(2)$ of Theorem~\ref{thm:Pc-extensions} are in one-to-one correspondence with $X$ itself. 
	The extensions described in $(3)$ correspond to the extremal points $\Ext D$ of $D$. 
	Since $D=\CoH(\Ext D)$ by Minkowski's theorem, cf.\ \cite[Corollary~1.4.5]{schneider:1993}, 
	this set is certainly not empty. However, it is also not too large, e.g., $\relint(D)$ 
	does not contain any extremal points and is a dense subset of $D$, cf.\ \cite[Theorem~1.1.14]{schneider:1993}.

	Case $(4)$ of the theorem is the most intriguing. To explain that the set of extensions 
	occurring in this way has a complicated structure, we consider two particular situations. 
	Namely, that the element $C$ is either closed or relatively open. 
	\begin{enumerate}
	\item Let $C\in\mathcal K(D)$ with more than one point. Then (by Lemma~\ref{lem:conv-body-2} and
		Lemma~\ref{lem:conv-body-3} in Section~\ref{sec:appendix-examples})
		\[
			\{A\in X\mid A\adh C\}=\{C\},\quad 
			I=\CoH\big(X\setminus\{C\}\big)=
			\begin{cases}
				X\setminus\{C\} &\hspace*{-3mm},\quad C\in\Ext\mathcal K(D),
				\\
				X &\hspace*{-3mm},\quad \text{otherwise}.
			\end{cases}
		\]
		Hence, $C$ is eligible for the construction in $(4)$ if and only if $C\in\Ext\mathcal K(D)$. 
		If $C$ is an extremal point there is a unique choice for the prime ideal $P$, namely $P=X\setminus\{C\}$. 
	\item Let $C\in\mathcal K(D)$ be relatively open with more than one point. 
		Then (by Lemma~\ref{lem:conv-body-1} in Section~\ref{sec:appendix-examples})
		\[
			\{A\in X\mid A\adh C\}=\{A\in X\mid \overline A=\overline C\},\quad 
			I=\CoH(\{A\in X\mid \overline A\neq\overline C\})
			.
		\]
		Assume that $\overline C\in\Ext\mathcal K(D)$. Then $I$ itself is a prime ideal, and we may choose $P=I$ in $(4)$. 
		There are also other possible choices for $P$. For example, $P=X\setminus\{\overline C\}$ is a prime ideal. 

		If $\overline C\not\in\Ext\mathcal K(D)$, we have no general result telling how large $I$ will be. 
		In some concrete low-dimensional examples, we saw that $I\neq X$ may happen and there are many 
		prime ideals admissible for $(4)$. 
	\end{enumerate}
	To summarize, we exhibited two families of different extensions: 
	1.\ A unique extension is associated with each extremal point $C$ of $\mathcal K(D)$. 
	Points $C\in\mathcal K(D)\setminus\Ext\mathcal K(D)$ are not eligible; 
	2.\ Every relatively open subset $C$ with $\overline C\in\Ext\mathcal K(D)$ is eligible and each such set 
	gives rise to several extensions. If $\overline C\not\in\Ext\mathcal K(D)$, examples indicate that $C$ can be 
	eligible and give rise to different extensions. 
\end{exa}

\begin{rem}\label{rem:conv-body-extr}
	Describing $\Ext\mathcal K(D)$ is an open problem in convex geometry. 
	$\Ext\mathcal K(D)$ must be large in the sense that $\CoH\Ext\mathcal K(D)$ is dense in $\mathcal K(D)$
	w.r.t.\ the Hausdorff metric.
	This fact is shown by considering $\mathcal K(D)$ as a compact convex subset of the space of continuous functions on 
	the $n$-dimensional sphere via passing to support functions (e.g.\ \cite[Theorems~1.7.1,1.8.11 and p.150]{schneider:1993}), 
	and applying the Krein-Milman theorem (e.g.\ \cite[Theorem~3.23]{rudin:1991}).

	The only situation we know where an explicit description of $\Ext\mathcal K(D)$ is possible is 
	\cite[Theorem]{grzaslewicz:1984} where $D$ is a strictly convex subset of $\mathbb R^2$ with nonempty interior. 
\end{rem}

When $D$ is a simplex (defined below), a bit more can be said about $\Ext\mathcal K(D)$. 

\begin{exa}\label{ex:dis-S-finite}
	Let $n\geq 2$, let $S$ be an $n$-element set, and consider $\mathcal K(\Dis S)$. 
	The algebra $\Dis S$ is isomorphic to the standard $(n-1)$-simplex 
	\[
		\Delta^{n-1}:=\big\{(t_1,\ldots,t_n)\in\mathbb R^n \mid t_i\geq 0,\sum_{i=1}^n t_i=1\big\}
		.
	\]
	The extremal points of $\Delta^{n-1}$ are its corner points, the canonical basis vectors $e_i$ in $\mathbb R^n$. 
	Hence a singleton $\{\xi\}$ belongs to $\Ext\mathcal K(\Dis S)$ if and only if $\xi$ is a Dirac measure $\delta_x$ concentrated 
	at one of the points $x$ of $S$. The set 
	\begin{equation}\label{eq:set-E}
		\mathcal E:=\{A\in\Ext\mathcal K(\Dis S) \mid A\text{ is not singleton}\}
	\end{equation}
	corresponds bijectively to the homothety classes of indecomposable convex bodies in $\mathbb R^{n-1}$ 
	which are not singletons (see Corollary~\ref{cor:Dis-3}, Lemma~\ref{lem:Dis-4} in Section~\ref{sec:appendix-examples}). 
	Let us point out that this property is particular for the algebra $\mathcal K(\Delta^{n-1})$.
	For example, it does not hold for $\mathcal K(D)$ when $D$ is a square in $\mathbb R^2$, 
	cf.\ \cite[Remark~2]{grzaslewicz:1984}.

	For $n=2$ and $n=3$ the set $\mathcal E$ is known explicitly; the case $n=2$ is trivial and the case $n=3$ is elaborated in
	\cite[Theorem~3.2.11]{schneider:1993}. 
	Concretely, for $n=2$ we have
	\[
		\Ext\mathcal K(\Dis\{1,2\})=
		\big\{\{\delta_1\},\{\delta_2\}\big\}\cup 
		\big\{\CoH\{\delta_1,\delta_2\}\big\} = \big\{\{\delta_1\},\{\delta_2\}, \Dis\{1,2\}\big\}
		.
	\]
	For $n=3$, $\Ext\mathcal K(\Dis\{1,2,3\})$ consists of 
	the singletons $\{\delta_1\},\{\delta_2\},\{\delta_3\}$, 
	the line segments connecting one Dirac measure with any convex combination of the other two, 
	and the triangles having at least one corner point in each of 
	$\CoH\{\delta_2,\delta_3\},\CoH\{\delta_1,\delta_3\},\CoH\{\delta_1,\delta_2\}$. 

	For $n\geq 4$, $\Ext\mathcal K(\Dis S)$ is dense in $\mathcal K(\Dis S)$, cf.\ \cite[Theorem~3.2.14]{schneider:1993}. 
\end{exa}

\begin{rem}\label{rem:dis-S-finite-extr}
	Giving description of indecomposable convex bodies in $\mathbb R^n$ with $n\geq 3$ is an open problem in convex geometry. 
	By Baire's category theorem most convex bodies are indecomposable (being 
a dense set). However, almost no
indecomposable bodies are explicitly known.,
	cf.\ the discussion \cite[p.153]{schneider:1993}. For polytopes, however, there are easy to check conditions for 
	indecomposability. For example, if all $2$-dimensional faces of a polytope are triangles, then it is indecomposable, 
	cf.\ \cite[Corollary~3.2.13]{schneider:1993}. 
\end{rem}

\section{Proofs of Geometric Arguments in Example~\ref{ex:conv-body} and Example~\ref{ex:dis-S-finite}}
\label{sec:appendix-examples}

In this section we provide evidence for the statements made in Example~\ref{ex:conv-body} and Example~\ref{ex:dis-S-finite}.

Example~\ref{ex:conv-body} relies on knowing the sets $\{A\in\CSS\Alg D\mid A\adh C\}$ when $C$ is either 
closed or relatively open, and on the fact that an extremal point of $\mathcal K(D)$ is automatically 
extremal in the larger algebra $\CSS\Alg D$. 
The next three results contain all the details. 

\begin{lem}\label{lem:conv-body-1}
	Let $D\in\mathcal K^n$. 
	\begin{enumerate}
	\item The closure operator $\overline{\rule{0pt}{5pt}.}:\CSS\Alg D\to\mathcal K(D)$ is a homomorphism of convex algebras. 
	\item Let $A,B\in\CSS\Alg D$. If $A\adh B$, then $\overline A=\overline B$. 
	\item Let $C\in\CSS\Alg D$ be relatively open. Then 
		\[
			\{A\in\CSS\Alg D\mid A\adh C\}=\{A\in\CSS\Alg D\mid \overline A=\overline C\}
			.
		\]
	\end{enumerate}
\end{lem}
\begin{proof}
\hfill
\begin{enumerate}
\item 
	Let $A,B\in\CSS\Alg D$, $p\in(0,1)$, and set $C=pA+\bar pB$. Continuity of linear operations ensures that 
	$p\overline A+\bar p\overline B\subseteq\overline C$. However, $\overline A$ and $\overline B$ are compact, 
	so also $p\overline A+\bar p\overline B$ is compact, hence, in particular, closed. We conclude that 
	$\overline C\subseteq p\overline A+\bar p\overline B$. 
\item
	If $A\adh B$, by 1., also $\overline A\adh\overline B$, i.e., $p\overline A+\bar p\overline B=\overline B$ for 
	all $p\in(0,1)$. This implies $\overline A=\lim_{p\to 1}(p\overline A+\bar p\overline B)=\overline B$, where 
	the limit is understood w.r.t.\ the Hausdorff metric. For the first equality recall that the operations are 
	continuous, cf.\ Remark~\ref{rem:imported-facts-geometry}.1.
\item
	The inclusion ``$\subseteq$'' holds by 2. For the reverse inclusion, let $A\in\CSS\Alg D$ with 
	$\overline A=\overline C$ be given. 
	Then $C=\relint(\overline C)=\relint(\overline A)=\relint A\subseteq A$, cf.\ \cite[Theorem~6.3]{rockafellar:1970}. 
	Thus $pA+\bar pC\supseteq C$. The inclusion $pA + \bar pC\subseteq C$ holds by \cite[Theorem~6.1]{rockafellar:1970}. 
\qedhere
\end{enumerate}
\end{proof}

\begin{lem}\label{lem:conv-body-2}
	Let $C\in\mathcal K^n$. 
	\begin{enumerate}
	\item If $A,B\in\CSS\mathbb R^n$, $p\in(0,1)$, with $A,B\subseteq C$ and $pA+\bar pB=C$, then $A=B=C$.
	\item ${\displaystyle \{A\in\CSS\mathbb R^n\mid A\adh C\}=\{C\}}$. 
	\end{enumerate}
\end{lem}
\begin{proof}
\hfill
\begin{enumerate}
\item
	Let $x\in\Ext C$, and choose $a\in A$, $b\in B$, with $x=pa+\bar pb$. Since $a,b\in C$, it 
	follows that $a=b=x$. We see that $\Ext C\subseteq A\cap B$. Using Minkowski's theorem, 
	cf.\ \cite[Corollary~1.4.5]{schneider:1993}, yields 
	$C=\CoH(\Ext C)\subseteq A\cap B$, from which $A=B=C$ follows. 
\item
	Let $a\in A$ and choose $c\in C$. Then, for each $p\in(0,1)$, the element $pa+\bar pc$ belongs to $C$. 
	Since $C$ is closed, 
	\[
		a=\lim_{p\to 1}(pa+\bar pc)\in C
		.
	\]
	Thus $A\subseteq C$, and 1.\ shows that $A=C$. 
\qedhere
\end{enumerate}
\end{proof}

\begin{lem}\label{lem:conv-body-3}
	Let $D\in\mathcal K^n$ and $C\in\mathcal K(D)$. 
	Then $C$ is an extremal point of $\mathcal K(D)$ if and only if it is an extremal point of $\CSS\Alg D$. 
\end{lem}
\begin{proof}
	Since $\mathcal K(D)$ is a subalgebra of $\CSS\Alg D$, we have $\Ext(\CSS\Alg D)\cap\mathcal K(D)\subseteq\Ext\mathcal K(D)$. 
	Assume $C\in\Ext\mathcal K(D)$ and that $C=pA+\bar pB$ with some $A,B\in\CSS\Alg D$ and $p\in(0,1)$. 
	Then $p\overline A+\bar p\overline B=\overline C=C$, and we obtain $\overline A=\overline B=C$. In particular, 
	$A,B\subseteq C$, and Lemma~\ref{lem:conv-body-2}.1 shows $A=B=C$. 
\end{proof}

We turn to the proof of the statement made in Example~\ref{ex:dis-S-finite}, that the set $\mathcal E$ from (\ref{eq:set-E}) 
corresponds to (classes of) indecomposable bodies. 

First, let us introduce some notation.
Let $\pi_i\colon\mathbb R^n\to\mathbb R$ denote the $i$-th projection, and set $\sigma=\sum_{i=1}^n\pi_i$. 
Denote by $e_i$ the $i$-th canonical basis vector of $\mathbb R^n$. 
Then the standard $n$-simplex can be written as 
\[
	\Delta^n=\CoH\{e_1,\ldots,e_n\}=\sigma^{-1}(\{1\})\cap\bigcap_{i=1}^n\pi_i^{-1}([0,1])
	.
\]
A central role is played by the subset of $\mathcal K(\Delta^n)$ of all bodies which touch each face of $\Delta^n$.

\begin{defi}\label{def:touching}
	Set $\mathcal L(\Delta^n)=\{A\in\mathcal K(\Delta^n)\mid \min\pi_i(A)=0,i=1,\ldots,n\}$.
\qedd
\end{defi}

\begin{lem}\label{lem:Dis-1}
	The set $\mathcal L(\Delta^n)$ contains no singletons. 
	It is a convex and extreme subset of $\mathcal K(\Delta^n)$. 
\end{lem}
\begin{proof}
	Assume $\{x\}\in\mathcal L(\Delta^n)$. Since $\min\pi_i(\{x\})=\pi_i(x)$, it follows that $x=0$ which contradicts 
	$\sigma(x)=\max\sigma(\{x\})=1$. 

	It holds that $\min\pi_i(pA+\bar pB)=p\min\pi_i(A)+\bar p\min\pi_i(B)$. The fact that $\mathcal L(\Delta^n)$ is 
	convex follows immediately. 
	To show that $\mathcal L(\Delta^n)$ is extreme, let $A\in\mathcal K(\Delta^n)\setminus\mathcal L(\Delta^n)$. 
	Choose $j\in\{1,\ldots,n\}$ with $\min\pi_j(A)>0$. Then, for each $p\in(0,1)$ and $B\in\mathcal K(\Delta^n)$ we have 
	$\min\pi_i(pA+\bar pB)\geq p\min\pi_i(A)>0$, and conclude that $pA+\bar pB\not\in\mathcal L(\Delta^n)$. 
\end{proof}

Next we prove that each extremal point of $\mathcal K(\Delta^n)$ which is not a singleton 
must touch each face of $\Delta^n$. 
This property fits our intuition, but surprisingly it only holds for the algebra $\mathcal K(\Delta^n)$. 
For example, it does not hold in $\mathcal K(D)$ for $D$ a square in $\mathbb R^2$, cf.\ \cite[Remark~2]{grzaslewicz:1984}. 

\begin{lem}\label{lem:Dis-2}
	We have $\big\{A\in\Ext\mathcal K(\Delta^n)\mid A\text{ is not singleton}\big\}\subseteq\mathcal L(\Delta^n)$. 
\end{lem}
\begin{proof}
	Let $A\in\mathcal K(\Delta^n)\setminus\mathcal L(\Delta^n)$ have more than one point. 
	Our aim is to show that $A$ can be written as a convex combination of two elements of $\mathcal K(\Delta^n)$ not both 
	equal to $A$.
	Choose $j\in\{1,\ldots,n\}$ with $\epsilon=\min\pi_j(A)>0$. Since $A$ is not a singleton, we have $\epsilon<1$. 
	Consider the convex map $f\colon x\mapsto(1-\epsilon)x+\epsilon e_j$ and its inverse 
	$f^{-1}\colon x\mapsto\frac 1{1-\epsilon}x-\frac\epsilon{1-\epsilon}e_j$. 
	Obviously, $f(\Delta^n)\subseteq\Delta^n$. We prove that $f^{-1}(A)\subseteq\Delta^n$. 
	Let $x\in A$, then 
	\begin{align*}
		& \sigma(f^{-1}(x))=\frac 1{1-\epsilon}-\frac\epsilon{1-\epsilon}=1,
		\\
		& \pi_i(f^{-1}(x))=\frac 1{1-\epsilon}\pi_i(x)\geq 0,\ i\neq j,\quad 
		\pi_j(f^{-1}(x))=\frac 1{1-\epsilon}\big(\pi_i(x)-\epsilon\big)\geq 0
		.
	\end{align*}
	Thus $f^{-1}(x)\in\Delta^n$. 

	Set $p=\frac 1{2-\epsilon}$. Then $p\in(0,1)$, $\bar p=\frac{1-\epsilon}{2-\epsilon}$, and 
	\[
		pf(x)+\bar pf^{-1}(y)=\Big(\frac{1-\epsilon}{2-\epsilon}x+\frac\epsilon{2-\epsilon}e_j\Big)+
		\Big(\frac 1{2-\epsilon}x-\frac\epsilon{2-\epsilon}e_j\Big)=px+\bar py
		.
	\]
	It follows that $pf(A)+\bar pf^{-1}(A)=A$. Since $\min\pi_j(f^{-1}(A))=0$, we have $f^{-1}(A)\neq A$.
\end{proof}

\begin{cor}\label{cor:Dis-3}
	We have $\Ext\mathcal K(\Delta^n)=\{\{e_1\},\ldots,\{e_n\}\}\,\cup\,\Ext\mathcal L(\Delta^n)$. The union is disjoint.
\end{cor}
\begin{proof}
	We already determined the extremal points of $\mathcal K(\Delta^n)$ which are singletons. 
	Clearly, $\mathcal L(\Delta^n)\cap\Ext\mathcal K(\Delta^n)\subseteq\Ext\mathcal L(\Delta^n)$. 
	Since $\mathcal L(\Delta^n)$ is an extremal set, $\Ext\mathcal L(\Delta^n)\subseteq\Ext\mathcal K(\Delta^n)$. 
	If $A\in\Ext\mathcal K(\Delta^n)$ is not a singleton, the above lemma says $A\in\mathcal L(\Delta^n)$. 
\end{proof}

To make the connection with indecomposable bodies, we pass to an isomorphic situation. 
Slightly overloading notation, let again $\pi_i:\mathbb R^{n-1}\to\mathbb R$ 
be the $i$-th projection, and $\sigma$ and $e_i$ be defined as above (with $n-1$ instead of $n$). 
Set 
\[
	D^{n-1}=\big\{(t_1,\ldots,t_n)\in\mathbb R^{n-1} \mid t_i\geq 0,\sum_{i=1}^{n-1} t_i\leq 1\big\}
	.
\]
Then $D^{n-1}=\CoH\{0,e_1,\ldots,e_{n-1}\}=\sigma^{-1}([0,1])\cap\bigcap_{i=1}^{n-1}\pi_i^{-1}([0,1])$. The standard simplex 
$\Delta^n$ is isomorphic to $D^{n-1}$ via the convex map $\phi$ taking the basis vector $e_i$ of $\mathbb R^n$ to the 
corresponding basis vector in $\mathbb R^{n-1}$ if $i<n$ and mapping $e_n$ to $0$. 
This map lifts in the natural (pointwise) way to the isomorphism $\Phi\colon A\mapsto\phi(A)$ of the convex algebra 
$\mathcal K(\Delta^n)$ onto $\mathcal K(D^{n-1})$.
The image of $\mathcal L(\Delta^n)$ under this isomorphism is 
\[
	\Phi\big(\mathcal L(\Delta^n)\big)=
	\big\{A\in\mathcal K(D^{n-1})\mid \max\sigma(A)=1,\ \min\pi_i(A)=0,i=1,\ldots,n-1\big\}
	,
\]
which we shall denote as $\mathcal L(D^{n-1})$. 
Being an isomorphism, $\Phi$ maps extremal points to extremal points, and we obtain 
\[
	\Phi\big(\Ext\mathcal L(\Delta^n)\big)=\Ext\mathcal L(D^{n-1})
	.
\]
Further note that also $\mathcal L(D^{n-1})$ contains no singletons. 

\begin{lem}\label{lem:Dis-4}
	\hspace*{0pt}
	\begin{enumerate}
	\item The set $\mathcal L(D^{n-1})$ is a complete system of representatives modulo homothety 
		of the set of convex bodies in $\mathbb R^{n-1}$ having more than one point. 
	\item Let $A\in\mathcal L(D^{n-1})$. Then $A\in\Ext\mathcal L(D^{n-1})$ if and only if 
		$A$ is indecomposable in $\mathcal K^{n-1}$. 
	\end{enumerate}
\end{lem}
\begin{proof}
\hfill
\begin{enumerate}
\item
	Let $A\in\mathcal K^{n-1}$ with more than one point, and set $t_i=\min\pi_i(A)$, $x=\sum_{i=1}^{n-1}t_ie_i$. 
	Since $A$ is not a singleton, $\max\sigma(A)>\sum_{i=1}^{n-1}t_i$. The body
	\[
		\Psi(A)=\big(\max\sigma(A)-\sum_{i=1}^{n-1}t_i\big)^{-1}\big(A-x\big)
	\]
	is homothetic to $A$ and belongs to $\mathcal L(D^{n-1})$. 

	Now assume that $s>0$, $x\in\mathbb R^{n-1}$, and that both $A$ and $sA+x$ belong to $\mathcal L(D^{n-1})$. 
	Then 
	\[
		0=\min\pi_i(sA+x)=s[\min\pi_i(A)]+\pi_i(x)=\pi_i(x)
		,
	\]
	whence $x=0$. Now $1=\max\sigma(sA)=s\max\sigma(A)=s$. 
\item
	Assume $A\in\mathcal L(D^{n-1})\setminus\Ext\mathcal L(D^{n-1})$, then $A$ has a representation 
	$A=pB+\bar pC$ with some $p\in(0,1)$ and $B,C\in\mathcal L(D^{n-1})$ where not both of $B,C$ equal $A$. 
	Then $pB$ and $\bar pC$ are not both homothetic to $A$ by 1., 
	and we conclude that $A$ is decomposable in $\mathcal K^{n-1}$. 

	To show the converse, we first establish the following: 
	If $B,C\in\mathcal K^{n-1}$ are not singletons, then there exists $p\in(0,1)$ with 
	\[
		\Psi(B+C)=p\Psi(B)+\bar p\Psi(C)
		.
	\]
	To see this, denote 
	\begin{align*}
		& t_i^B=\min\pi_i(B),s^B=\max\sigma(B)-\sum_{i=1}^{n-1}t_i^B,\quad 
		t_i^C=\min\pi_i(C),s^C=\max\sigma(C)-\sum_{i=1}^{n-1}t_i^C,
		\\
		& t_i=\min\pi_i(B+C),s=\max\sigma(B+C)-\sum_{i=1}^{n-1}t_i.
	\end{align*}
	Then $t_i=t_i^B+t_i^C$ and $s=s^B+s^C$, which gives
	\[
		\Psi(B+C)=\frac 1{s^B+s^C}\Big[(B+C)+\sum_{i=1}^{n-1}(t_i^B+t_i^C)\Big]
		=\frac{s^B}{s^B+s^C}\Psi(B)+\frac{s^C}{s^B+s^C}\Psi(C)
		.
	\]
	Since $B$ and $C$ are not singletons, $s^B$ and $s^C$ are both nonzero. Thus $p=\frac{s^B}{s^B+s^C}\in(0,1)$. 

	Now let $A\in\mathcal L(D^{n-1})$, and assume that $A=B+C$ with $B,C\in\mathcal K^{n-1}$ not both 
	homothetic to $A$. Then 
	\[
		A=\Psi(A)=p\Psi(B)+\bar p\Psi(C)
		,
	\]
	and not both of $\Psi(B)$ and $\Psi(C)$ are equal to $A$ by 1.
\qedhere
\end{enumerate}
\end{proof}

Putting together Corollary~\ref{cor:Dis-3} and Lemma~\ref{lem:Dis-4}, we see that indeed the set $\mathcal E$ from (\ref{eq:set-E}) 
corresponds bijectively to the homothety classes of indecomposable convex bodies in $\mathbb R^{n-1}$ 
which are not singletons.


\vspace{8mm}
\section{Conclusions}
\label{sec:conc}

We have studied the possibility of extending a convex algebra by a single
element. We have proven that many different extensions are possible of which
only one gives rise to a functor on $\EM(\Dis)$. We have described all
extensions of $\Alg D_S$, the free convex algebra of probability distributions
over a set $S$, and of $\CSS \Alg D$, the convex algebra of convex subsets of a particular kind of
convex subset of a vector space. As a consequence of the latter result, we have
described all extensions of $\CSS \Alg D_S$ used for modelling probabilistic
automata.

It would be interesting to investigate whether the methods developed here could be useful in the study of Eilenberg-Moore algebras 
of the Giry monad on measurable spaces, or on subcategories of measurable spaces like Polish or analytic spaces.


\vspace{1cm}
\begin{thebibliography}{10}

\bibitem{BSV04:tcs}
Falk Bartels, Ana Sokolova, and Erik de~Vink.
\newblock A hierarchy of probabilistic system types.
\newblock {\em Theoretical Computer Science}, 327:3--22, 2004.

\bibitem{BSS17}
Filippo Bonchi, Alexandra Silva, and Ana Sokolova.
\newblock The power of convex algebras.
\newblock In {\em Proc.~CONCUR}, pages 23:1--23:18, LIPIcs~85, 2017.

\bibitem{boerger.kemper:1996}
Reinhard B\"orger and Ralf Kemper.
\newblock A cogenerator for preseparated superconvex spaces.
\newblock {\em Appl. Categ. Structures}, 4(4):361--370, 1996.

\bibitem{CPP09}
Pablo~Samuel Castro, Prakash Panangaden, and Doina Precup.
\newblock Equivalence relations in fully and partially observable {M}arkov
  decision processes.
\newblock In {\em Proc.~IJCAI}, pages 1653--1658, 2009.

\bibitem{CR11}
Silvia Crafa and Francesco Ranzato.
\newblock A spectrum of behavioral relations over {L}{T}{S}s on probability
  distributions.
\newblock In {\em Proc.~CONCUR}, pages 124--139. LNCS~6901, 2011.

\bibitem{DengH13}
Yuxin Deng and Matthew Hennessy.
\newblock On the semantics of {M}arkov automata.
\newblock {\em Inf. Comput.}, 222:139--168, 2013.

\bibitem{DGHM08}
Yuxin Deng, Rob~J. van Glabbeek, Matthew Hennessy, and Carroll Morgan.
\newblock Characterising testing preorders for finite probabilistic processes.
\newblock {\em LMCS}, 4(4), 2008.

\bibitem{DengGHM09}
Yuxin Deng, Rob~J. van Glabbeek, Matthew Hennessy, and Carroll Morgan.
\newblock Testing finitary probabilistic processes.
\newblock In {\em Proc.~{CONCUR}}, pages 274--288. LNCS~5710, 2009.

\bibitem{doberkat:2006}
Ernst-Erich Doberkat.
\newblock Eilenberg-{M}oore algebras for stochastic relations.
\newblock {\em Inform. and Comput.}, 204(12):1756--1781, 2006.

\bibitem{doberkat:2008}
Ernst-Erich Doberkat.
\newblock Erratum and addendum: {E}ilenberg-{M}oore algebras for stochastic
  relations.
\newblock {\em Inform. and Comput.}, 206(12):1476--1484, 2008.

\bibitem{flood:1981}
Joe Flood.
\newblock Semiconvex geometry.
\newblock {\em J. Austral. Math. Soc. Ser. A}, 30(4):496--510, 1980/81.

\bibitem{fritz:2015}
Tobias Fritz.
\newblock {Convex spaces I: Definition and Examples}.
\newblock arXiv:0903.5522v3 [math.MG].

\bibitem{grzaslewicz:1984}
Ryszard Grza\'slewicz.
\newblock Extreme convex sets in {${\mathbb R}^2$}.
\newblock {\em Arch. Math. (Basel)}, 43(4):377--380, 1984.

\bibitem{gudder.schroeck:1980}
Stanley Gudder and Franklin Schroeck.
\newblock Generalized convexity.
\newblock {\em SIAM J. Math. Anal.}, 11(6):984--1001, 1980.

\bibitem{HermannsKK14}
Holger Hermanns, Jan Krc{\'{a}}l, and Jan Kret{\'{\i}}nsk{\'{y}}.
\newblock Probabilistic bisimulation: Naturally on distributions.
\newblock In {\em Proc.~CONCUR}, LNCS~8704, pages 249--265, 2014.

\bibitem{jacobs:2010}
Bart Jacobs.
\newblock Convexity, duality and effects.
\newblock In {\em Theoretical computer science}, volume 323 of {\em IFIP Adv.
  Inf. Commun. Technol.}, pages 1--19. Springer, 2010.

\bibitem{JacobsSS15}
Bart Jacobs, Alexandra Silva, and Ana Sokolova.
\newblock Trace semantics via determinization.
\newblock {\em J. Comput. Syst. Sci.}, 81(5):859--879, 2015.

\bibitem{jacobs.westerbaan.westerbaan:2015}
Bart Jacobs, Bas Westerbaan, and Bram Westerbaan.
\newblock States of convex sets.
\newblock In {\em Proc.~FOSSACS}, pages 87--101. LNCS~9034, 2015.

\bibitem{kneser:1952}
Hellmuth Kneser.
\newblock Konvexe {R}\"aume.
\newblock {\em Arch. Math.}, 3:198--206, 1952.

\bibitem{Mio14}
Matteo Mio.
\newblock Upper-expectation bisimilarity and {\l}ukasiewicz {$\mu$}-calculus.
\newblock In {\em Proc.~FOSSACS}, pages 335--350. LNCS~8412, 2014.

\bibitem{neumann:1970}
Walter~D. Neumann.
\newblock On the quasivariety of convex subsets of affine spaces.
\newblock {\em Arch. Math. (Basel)}, 21:11--16, 1970.

\bibitem{PS07}
Augusto Parma and Roberto Segala.
\newblock Logical characterizations of bisimulations for discrete probabilistic
  systems.
\newblock In {\em Proc. FOSSACS}, pages 287--301. LNCS~4423, 2007.

\bibitem{pumpluen.roehrl:1995}
Dieter Pumpl\"un and Helmut R\"ohrl.
\newblock Convexity theories. {IV}. {K}lein-{H}ilbert parts in convex modules.
\newblock {\em Appl. Categ. Structures}, 3(2):173--200, 1995.

\bibitem{rockafellar:1970}
R.~Tyrrell Rockafellar.
\newblock {\em Convex analysis}.
\newblock Princeton Mathematical Series, No. 28. Princeton University Press,
  Princeton, N.J., 1970.

\bibitem{roehrl:1994}
Helmut R\"ohrl.
\newblock Convexity theories. 0. {F}oundations.
\newblock {\em Appl. Categ. Structures}, 2(1):13--43, 1994.

\bibitem{rudin:1991}
Walter Rudin.
\newblock {\em Functional analysis}.
\newblock International Series in Pure and Applied Mathematics. McGraw-Hill
  Inc., New York, second edition edition, 1991.

\bibitem{schneider:1993}
Rolf Schneider.
\newblock {\em Convex bodies: the {B}runn-{M}inkowski theory}, volume~44 of
  {\em Encyclopedia of Mathematics and its Applications}.
\newblock Cambridge University Press, Cambridge, 1993.

\bibitem{Seg95:thesis}
Roberto Segala.
\newblock {\em Modeling and verification of randomized distributed real-time
  systems}.
\newblock PhD thesis, MIT, 1995.

\bibitem{SL94}
Roberto Segala and Nancy Lynch.
\newblock Probabilistic simulations for probabilistic processes.
\newblock In {\em Proc.\ CONCUR}, pages 481--496. LNCS~836, 1994.

\bibitem{semadeni:1973}
Zbigniew Semadeni.
\newblock {\em Monads and their {E}ilenberg-{M}oore algebras in functional
  analysis}.
\newblock Queen's University, Kingston, Ont., 1973.
\newblock Queen's Papers in Pure and Applied Mathematics, No. 33.

\bibitem{SBBR10}
Alexandra Silva, Filippo Bonchi, Marcello Bonsangue, and Jan Rutten.
\newblock Generalizing the powerset construction, coalgebraically.
\newblock In {\em Proc. FSTTCS}, pages 272--283. LIPIcs~8, 2010.

\bibitem{Sokolova11}
Ana Sokolova.
\newblock Probabilistic systems coalgebraically: {A} survey.
\newblock {\em Theoretical Computer Science}, 412(38):5095--5110, 2011.

\bibitem{SW17}
Ana Sokolova and Harald Woracek.
\newblock Termination in convex sets of distributions.
\newblock In {\em Proc. CALCO}, to appear. LIPIcs, 2017.


\bibitem{stone:1949}
Marshall~Harvey Stone.
\newblock Postulates for the barycentric calculus.
\newblock {\em Ann. Mat. Pura Appl. (4)}, 29:25--30, 1949.

\bibitem{swirszcz:1974}
Tadeusz \'Swirszcz.
\newblock Monadic functors and convexity.
\newblock {\em Bull. Acad. Polon. Sci. S\'er. Sci. Math. Astronom. Phys.},
  22:39--42, 1974.

\bibitem{VR99:tcs}
Erik~de Vink and Jan Rutten.
\newblock Bisimulation for probabilistic transition systems: a coalgebraic
  approach.
\newblock {\em Theoretical Computer Science}, 221:271--293, 1999.

\bibitem{wickenhaeuser:1988}
Andreas Wickenh\"auser.
\newblock Positively convex spaces I.
\newblock {\em Seminarberichte, Hagen}, 30:119--172, 1988.

\end{thebibliography}
\end{document}